\documentclass{article}
\usepackage{graphics,amsmath, amsthm, amssymb,
latexsym,color}
\newtheorem{proposition}{Proposition}[subsection]
\newtheorem{theorem}{Theorem}[section]
\newtheorem{lemma}{Lemma}[subsection]

\newcommand{\nabb}{\mbox{$\nabla \mkern-13mu /$\,}}
\newcommand{\Trb}{\mbox{$\triangle \mkern-13mu /$\,}}

\def\pa{\partial}

\begin{document}
\title{The wave equation on
Schwarzschild-de Sitter spacetimes}
\author{Mihalis Dafermos\thanks{University of Cambridge,
Department of Pure Mathematics and Mathematical Statistics,
Wilberforce Road, Cambridge CB3 0WB United Kingdom}
\and
 Igor Rodnianski\thanks{Princeton University,
Department of Mathematics, Fine Hall, Washington Road,
Princeton, NJ 08544 United States}
}
\maketitle
\begin{abstract}
We consider solutions to the linear wave equation $\Box_g\phi=0$ on a non-extremal
maximally extended
Schwarzschild-de Sitter spacetime
 arising from arbitrary smooth initial data
 prescribed on an arbitrary Cauchy hypersurface. (In particular, no symmetry
 is assumed on initial data, and the support of the solutions may
 contain the sphere of bifurcation of the black/white hole horizons and the
 cosmological horizons.) We prove that
 in the region bounded by a set of black/white hole horizons and cosmological horizons, 
solutions $\phi$ converge pointwise to a constant
faster than any given polynomial rate, where the decay is measured
with respect to natural future-directed advanced and retarded time coordinates.
We also give such uniform decay bounds for the energy associated to the Killing field
as well as for the energy measured by local observers crossing the event horizon.
The results in particular include decay rates along the horizons themselves.
Finally, we discuss the relation of these results to previous heuristic
 analysis of Price and Brady et al. 
\end{abstract}
\tableofcontents
\section{Introduction}
The introduction of a positive cosmological constant
 in the Einstein equations of general relativity gives
rise to a wide variety of new interesting solution spacetimes, in particular, spacetimes containing both
``black hole'' and ``cosmological'' regions. As in the case of black-hole spacetimes with vanishing cosmological constant,
the stability of these spacetimes as solutions
to the Einstein equations is a fundamental
open problem of gravitational physics. Yet even the simplest questions concerning
the behaviour of linear waves on such spacetime backgrounds today remain unanswered.
In this paper, we initiate in the above context the mathematical
study of decay for solutions to the linear wave equation.

The simplest family of black-hole spacetimes with positive cosmological constant is
the so-called \emph{Schwarzschild-de
Sitter} family. If the cosmological constant $\Lambda>0$ is considered fixed, this 
is a $1$-parameter family of solutions $(\mathcal{M},g)$ 
to the \emph{Einstein vacuum equations}
\begin{equation}
\label{Eeq}
R_{\mu\nu}-\frac12g_{\mu\nu}R=-\Lambda g_{\mu\nu},
\end{equation}
with parameter $M$, called the \emph{mass}.
We shall consider only the non-extremal black-hole case, corresponding
to parameter values 
\begin{equation}
\label{NEBH}
0<M<\frac1{3\sqrt\Lambda}.
\end{equation}

As with the Schwarzschild family,
the first manifestation of the Schwarzschild-de Sitter family of solutions
was an expression for the metric  in local coordinates, in this case
first published in 1918 by Kottler~\cite{kottler}, and independently by Weyl~\cite{Weyl}, in the 
form 
\begin{equation}
\label{original}
-\left(1-\frac{2M}r-\frac13\Lambda r^2\right)dt^2
+\left(1-\frac{2M}r-\frac13\Lambda r^2\right)^{-1}dr^2
+r^2d\sigma_{{\mathbb S}^2}.
\end{equation}
Here $d\sigma_{{\mathbb S}^2}$ denotes the standard metric
on the unit $2$-sphere.
The global structure of maximal spherically symmetric
vacuum extensions of such
metrics was only understood much later~\cite{brandon,lake,GibHawk} based on the methods of
formal Penrose diagrams introduced 
by B.~Carter.
In fact, maximally extended
spherically symmetric vacuum spacetimes $(\mathcal{M},g)$ 
with various different topologies 
can be constructed, all of which equally well merit
the name ``Schwarzschild-de Sitter with parameter $M$ and
cosmological constant $\Lambda$''. 
Such solutions $(\mathcal{M},g)$ all share the 
property that the universal 
cover $\tilde{\mathcal{Q}}$ of the $2$-dimensional 
Lorentzian quotient $\mathcal{Q}=\mathcal{M}/SO(3)$ 
consists of an infinite chain of regions as depicted
in the Penrose diagram below:
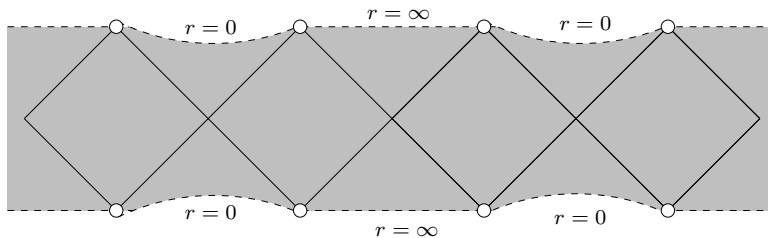
\[
\begin{picture}(0,0)%
\includegraphics{desit.pstex}%
\end{picture}%
\setlength{\unitlength}{2763sp}%
\begingroup\makeatletter\ifx\SetFigFont\undefined%
\gdef\SetFigFont#1#2#3#4#5{%
  \reset@font\fontsize{#1}{#2pt}%
  \fontfamily{#3}\fontseries{#4}\fontshape{#5}%
  \selectfont}%
\fi\endgroup%
\begin{picture}(6849,2164)(2389,-4694)
\put(5626,-2686){\makebox(0,0)[lb]{\smash{\SetFigFont{8}{9.6}{\rmdefault}{\mddefault}{\updefault}{\color[rgb]{0,0,0}$r=\infty$}%
}}}
\put(5701,-4636){\makebox(0,0)[lb]{\smash{\SetFigFont{8}{9.6}{\rmdefault}{\mddefault}{\updefault}{\color[rgb]{0,0,0}$r=\infty$}%
}}}
\put(3991,-2821){\makebox(0,0)[lb]{\smash{\SetFigFont{8}{9.6}{\rmdefault}{\mddefault}{\updefault}{\color[rgb]{0,0,0}$r=0$}%
}}}
\put(7351,-2791){\makebox(0,0)[lb]{\smash{\SetFigFont{8}{9.6}{\rmdefault}{\mddefault}{\updefault}{\color[rgb]{0,0,0}$r=0$}%
}}}
\put(3991,-4486){\makebox(0,0)[lb]{\smash{\SetFigFont{8}{9.6}{\rmdefault}{\mddefault}{\updefault}{\color[rgb]{0,0,0}$r=0$}%
}}}
\put(7306,-4531){\makebox(0,0)[lb]{\smash{\SetFigFont{8}{9.6}{\rmdefault}{\mddefault}{\updefault}{\color[rgb]{0,0,0}$r=0$}%
}}}
\end{picture}

\]
The results of this paper do not depend on the topology, and
for definiteness, one may assume in what follows that the name ``Schwarzschild
de-Sitter'' and the notation
$(\mathcal{M},g)$ refer to the spacetime with
quotient precisely the universal cover depicted above.

It is then the wave equation
\begin{equation}
\label{waveeq}
\Box_g\phi=0
\end{equation}
 on this background $(\mathcal{M},g)$
whose mathematical study we wish to initiate here. 
There is already a rich body of heuristic work on this problem in the physics
literature. (See Section~\ref{discussion} below for a discussion.)
The motivation for
the study of $(\ref{waveeq})$ in the present context is multifold. In particular, as in the case 
of vanishing cosmological constant,
studied in our previous~\cite{dr3}, we believe that proving bounds on decay
rates for solutions to  $(\ref{waveeq})$ is a first step
to a mathematical understanding of non-linear stability problems for spacetimes containing
black holes, that is to say, to the problem of stability in the context of the dynamics
of $(\ref{Eeq})$. For a more detailed discussion,
we refer the reader to the introductory remarks
of~\cite{dr3}.


\subsection{The initial value problem for the wave equation}
We are interested in solutions of $(\ref{waveeq})$ arising from suitably regular initial data
prescribed on a Cauchy surface $\Sigma$ of $\mathcal{M}$. For future applications
to non-linear stability problems, it is crucial that all assumptions have a natural geometric
interpretation independent of special coordinate systems.
Moreover,
our primary concern in this paper is the region $\mathcal{D}$ bounded
by a set of black/white hole horizons $\mathcal{H}^+\cup
\mathcal{H}^-$ and cosmological horizons $\overline{\mathcal{H}}^+\cup\overline
{\mathcal{H}}^-$:
\begin{equation}
\label{Ddef}
\mathcal{D}= {\rm clos}\big(J^-(\mathcal{H}^+\cup \overline{\mathcal{H}}^+)
\cap J^+(\mathcal{H}^-\cup\overline{\mathcal{H}}^-)\big)
\end{equation}
as depicted below\footnote{We employ in this paper the standard notation of
Lorentzian geometry (e.g.~$J^+$, $J^-$, etc.), and Penrose diagrams. See~\cite{he:lssst}.}
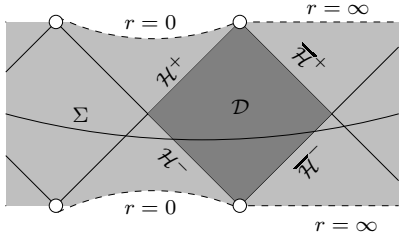
\[
\begin{picture}(0,0)%
\includegraphics{desitnewwo.pstex}%
\end{picture}%
\setlength{\unitlength}{2763sp}%
\begingroup\makeatletter\ifx\SetFigFont\undefined%
\gdef\SetFigFont#1#2#3#4#5{%
  \reset@font\fontsize{#1}{#2pt}%
  \fontfamily{#3}\fontseries{#4}\fontshape{#5}%
  \selectfont}%
\fi\endgroup%
\begin{picture}(3549,2164)(2914,-4694)
\put(5626,-2686){\makebox(0,0)[lb]{\smash{{\SetFigFont{8}{9.6}{\rmdefault}{\mddefault}{\updefault}{\color[rgb]{0,0,0}$r=\infty$}%
}}}}
\put(5701,-4636){\makebox(0,0)[lb]{\smash{{\SetFigFont{8}{9.6}{\rmdefault}{\mddefault}{\updefault}{\color[rgb]{0,0,0}$r=\infty$}%
}}}}
\put(3991,-2821){\makebox(0,0)[lb]{\smash{{\SetFigFont{8}{9.6}{\rmdefault}{\mddefault}{\updefault}{\color[rgb]{0,0,0}$r=0$}%
}}}}
\put(3991,-4486){\makebox(0,0)[lb]{\smash{{\SetFigFont{8}{9.6}{\rmdefault}{\mddefault}{\updefault}{\color[rgb]{0,0,0}$r=0$}%
}}}}
\put(4276,-3961){\rotatebox{315.0}{\makebox(0,0)[lb]{\smash{{\SetFigFont{8}{9.6}{\rmdefault}{\mddefault}{\updefault}{\color[rgb]{0,0,0}$\mathcal{H}^-$}%
}}}}}
\put(3526,-3661){\makebox(0,0)[lb]{\smash{{\SetFigFont{8}{9.6}{\rmdefault}{\mddefault}{\updefault}{\color[rgb]{0,0,0}$\Sigma$}%
}}}}
\put(5476,-3061){\rotatebox{315.0}{\makebox(0,0)[lb]{\smash{{\SetFigFont{8}{9.6}{\rmdefault}{\mddefault}{\updefault}{\color[rgb]{0,0,0}$\overline{\mathcal{H}}^+$}%
}}}}}
\put(4351,-3361){\rotatebox{45.0}{\makebox(0,0)[lb]{\smash{{\SetFigFont{8}{9.6}{\rmdefault}{\mddefault}{\updefault}{\color[rgb]{0,0,0}$\mathcal{H}^+$}%
}}}}}
\put(4951,-3586){\makebox(0,0)[lb]{\smash{{\SetFigFont{8}{9.6}{\rmdefault}{\mddefault}{\updefault}{\color[rgb]{0,0,0}$\mathcal{D}$}%
}}}}
\put(5626,-4186){\rotatebox{45.0}{\makebox(0,0)[lb]{\smash{{\SetFigFont{8}{9.6}{\rmdefault}{\mddefault}{\updefault}{\color[rgb]{0,0,0}$\overline{\mathcal{H}}^-$}%
}}}}}
\end{picture}%

\]
By causality, the global behaviour of $\phi$ in $\mathcal{D}$ can be understood independently
of the behaviour near $r=0$ and $r=\infty$. The behaviour in say
$\mathcal{D}\cap J^+(\Sigma)$ 
is completely
determined by the behaviour of appropriate initial data 
on $\Sigma\cap J^-(\mathcal{D})$.\footnote{Note, as depicted, that $\Sigma\cap J^-(\mathcal{D})$ is
not necessarily $\Sigma \cap \mathcal{D}$.}
We review briefly in the next paragraph
the solvability and domain of dependence property
for the initial value problem for $(\ref{waveeq})$.

Let $\Sigma\subset \mathcal{M}$ be a smooth Cauchy surface  and
let $n^\mu$ denote the future-directed unit normal of $\Sigma$.
For $s\ge 1$, let $\varphi$ be an $H^s_{\rm loc}(\Sigma)$ function
and 
$\dot\varphi:\Sigma \to \mathbb R$ an $H^{s-1}_{\rm loc}(\Sigma)$ function.
Then 
there exists a unique global solution $\phi:\mathcal{M}\to \mathbb R$ of $\Box_g\phi=0$
such that for all smooth spacelike hypersurfaces $\tilde{\Sigma}$ with future directed unit normal
$\tilde{n}$, $\phi|_{\Sigma'}\in H^{s}_{\rm loc}$, $(\tilde{n}\phi)|_{\Sigma'}
\in H^{s-1}_{\rm loc}$, and $\phi|_{\Sigma}=\varphi$, $n\phi|_{\Sigma}=\dot\varphi$.
Moreover, if
$K\subset \mathcal{M}$ is closed and 
$\phi_1$, $\phi_2$ are two such solutions corresponding to data
$(\varphi_1,\dot\varphi_1)$, $(\varphi_2,\dot\varphi_2)$ such that
$\varphi_1|_{\Sigma'\cap K}=\varphi_2|_{\Sigma'\cap K}$,
$\dot\varphi_1|_{\Sigma'\cap K}=\dot\varphi_2|_{\Sigma'\cap K}$, then
$\phi_1=\phi_2$ on $\mathcal{M}\setminus (J^+(\Sigma\setminus K)\cup J^-(\Sigma\setminus
K))$. 
In particular, setting $K=J^-(\mathcal{D})$, we obtain that $\phi_1=\phi_2$
on $J^+(\Sigma)\cap \mathcal{D}$.

\subsection{The main theorem}
\subsubsection{Norms on initial data}
Since our results will be quantitative, 
we need to introduce relevant norms on the compact manifold
with boundary $\Sigma\cap J^-(\mathcal{D})$. 
Let $\|\cdot\|$ denote the Riemannian $L^2$ norm
on $\Sigma\cap J^-(\mathcal{D})$.  This induces a norm on sections of 
the tangent bundle, a norm we will denote also by $\|\cdot\|$. If 
$\varphi\in H^1_{\rm loc}(\Sigma)$, 
 $\dot\varphi\in L^2_{\rm loc}(\Sigma)$, then
 let us denote by $\phi$ the unique solution of $\Box_g\phi=0$ corresponding
 to initial data $(\varphi,\dot\varphi)$. 

Let us define now for all real $s\ge 0$
the quantity
\begin{eqnarray}
\label{altydef}
{\bf E}_s(\varphi,\dot\varphi)
			&\doteq&			\|\nabla_{\Sigma}\varphi\|^2+ \|\dot\varphi\|^2+\sum_{\ell\ge 1}
			 r^{2s}
				\ell^{2s} ||\nabla_{\Sigma}\phi_\ell ||^2+r^{2s} \ell^{2s} ||n\phi_{\ell}||^2,
\end{eqnarray}
where $\phi_\ell$ denotes the projection of $\phi$ to the $\ell$'th eigenspace of
$\Trb$, i.e.~the $\ell$'th spherical harmonic of $\phi$. 
The function $r$ is discussed in Section~\ref{cordsec}. If 
$\Sigma$ itself is spherically symmetric, then we may replace $\phi_\ell$, $n\phi_\ell$
be $\varphi_\ell$ and $\dot\varphi_\ell$, and the above expression is a sum of
integrals on initial data.
For general $\Sigma$, 
a sufficient condition for the finiteness of $(\ref{altydef})$ is
that $\varphi\in H^{s+1}_{\rm loc}(\Sigma)$,
$\dot\varphi\in H^{s}_{\rm loc}(\Sigma)$.

In the case $m\ge 0$ an integer, we can characterize ${\bf E}_m$
geometrically as follows.
Let $\Omega_i$, $i=1,\ldots,3$ denote a basis of Killing fields generating
the Lie algebra $so(3)$ associated to the spherical symmetry of $(\mathcal{M},g)$. 
We call $\Omega_i$ \emph{angular momentum operators}. It easily follows that
\begin{eqnarray*}
{\bf E}_m(\varphi,\dot\varphi)	\sim
								\sum_{p_1,\ldots p_{m-1}=0,1}\sum_{1\le i_1, \ldots
								i_{m-1}\le 3}&\|\nabla_\Sigma (\Omega_{i_1}^{p_1}
								\cdots\Omega_{i_{m-1}}^{p_{m-1}}\phi)\|^2\\
						&\hbox{}+
								\|n(\Omega_{i_1}^{p_1}
								\cdots\Omega_{i_{m-1}}^{p_{m-1}}\phi) \|^2.
\end{eqnarray*}
Again, if $\Sigma$ itself
is spherically symmetric, we may replace $\phi$ with $\varphi$ in the first
term, and remove the $n$ from the second, replacing $\phi$ with $\dot\varphi$.

\subsubsection{First statement of the theorem}
The main result of this paper is contained in the following
\begin{theorem}
\label{gendata}
Let $(\mathcal{M},g)$ denote the Schwarzschild-de Sitter spacetime with parameter
$M$ and cosmological constant $\Lambda$ satisfying $(\ref{NEBH})$ and let
$\Sigma$ be a Cauchy surface for $\mathcal{M}$.
Let $\mathcal{D}\subset \mathcal{M}$ denote a region
as defined in $(\ref{Ddef})$ and let $s\ge 0$.
Then, there exist constants $C_s$ depending only on $s$, $M$, $\Lambda$,  
and the geometry of $\Sigma \cap J^-(\mathcal{D})$ such that
for all solutions $\phi$ of the wave equation $\Box_g\phi=0$ on
$\mathcal{M}$ such that ${\bf E}_s(\varphi, \dot\varphi)$ is finite, where
$\varphi\doteq\phi|_\Sigma$, $\dot\varphi \doteq n\phi|_{\Sigma}$,
and for 
all achronal hypersurfaces $\Sigma'\subset\mathcal{D}\cap J^+(\Sigma')$,
the bound
\begin{equation}
\label{fbxwris}
\int_{\Sigma'}T_{\mu\nu}(\phi)T^\mu n^\nu \le C_s \,
{\bf E}_s(\varphi, \dot\varphi)(v_+(\Sigma')^{-s}+u_+(\Sigma')^{-s})
\end{equation}
holds, 
where $u$ and $v$ denote fixed Eddington-Finkelstein 
advanced and retarded coordinates\footnote{See
Section~\ref{cordsec}. Although these coordinates are only defined in $\mathcal{D}^o$,
the statements $(\ref{fbxwris})$, $(\ref{ptwise1})$ can be interpreted in all of $\mathcal{D}$ in view of
conventions $(\ref{conve1})$--$(\ref{conve4})$.},
$u_+\doteq\max \{u,1\}$, $u_+(\Sigma)\doteq\inf_{x\in\Sigma} u_+(x)$, etc, 
$T^\mu$ denotes the Killing field coinciding in the interior of $\mathcal{D}$ with 
$\frac{\partial}{\partial t}$, $T_{\mu\nu}(\phi)$ denotes the standard
energy-momentum tensor, and $n^\nu$ is the future-directed
unit normal wherever $\Sigma'$ is spacelike, in which case the integral is
taken with measure the induced volume form.\footnote{A correct interpretation of
$n^\mu$ and the measure of integration for general achronal $\Sigma'$
can be derived by a limiting procedure.}

In addition, $(\ref{fbxwris})$ 
holds if  $T^\mu$ is replaced by
the vector field $N^\mu$ defined in Section~\ref{SMEsec}.
If  $s>1$, then the pointwise bound 
\begin{equation}
\label{ptwise1}
|\phi-\underline\phi|\le C_{s}\,{\bf E}^{\frac 12}_s(\varphi, \dot\varphi)\Big(v_+^{\frac{-s+1}2}+u_+^{\frac{-s+1}2}\Big)
\end{equation}
holds
in $J^+(\Sigma)\cap \mathcal{D}$,
where $\underline\phi$ is a constant satisfying
\[
|\underline{\phi}| \le  \sup_{x\in\Sigma}|\phi(x)| + C_0\, {\bf E}^{\frac 12}_0(\varphi, \dot\varphi).
\]
\end{theorem}
{\bf
In particular, the theorem applies to arbitrary smooth initial data $\varphi\in C^\infty(\Sigma)$, $\dot\varphi\in
C^\infty(\Sigma)$, where $s$ can be taken arbitrarily large.}
There are no unphysical
assumptions regarding vanishing of $\phi$ at the \emph{sphere of bifurcation} of the horizons,
i.e.~at the sets
$\mathcal{H}^+\cap \mathcal{H}^-$ and $\overline{\mathcal{H}}^+\cap \overline{\mathcal{H}}^-$.
The decay rates $(\ref{ptwise1})$, $(\ref{fluxbound})$ are uniform, i.e.~they
hold up to and including the horizons,
setting $u_+=\infty$ or $v_+=\infty$. 
In particular, $\Sigma'$
in $(\ref{fluxbound})$ can be taken (as depicted below)
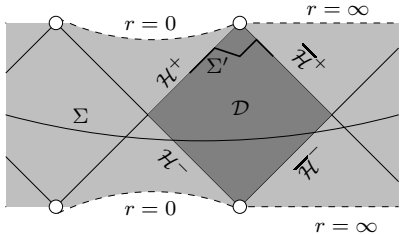
\[
\begin{picture}(0,0)%
\includegraphics{desitnew.pstex}%
\end{picture}%
\setlength{\unitlength}{2763sp}%
\begingroup\makeatletter\ifx\SetFigFont\undefined%
\gdef\SetFigFont#1#2#3#4#5{%
  \reset@font\fontsize{#1}{#2pt}%
  \fontfamily{#3}\fontseries{#4}\fontshape{#5}%
  \selectfont}%
\fi\endgroup%
\begin{picture}(3549,2164)(2914,-4694)
\put(5701,-4636){\makebox(0,0)[lb]{\smash{{\SetFigFont{8}{9.6}{\rmdefault}{\mddefault}{\updefault}{\color[rgb]{0,0,0}$r=\infty$}%
}}}}
\put(5626,-2686){\makebox(0,0)[lb]{\smash{{\SetFigFont{8}{9.6}{\rmdefault}{\mddefault}{\updefault}{\color[rgb]{0,0,0}$r=\infty$}%
}}}}
\put(3991,-4486){\makebox(0,0)[lb]{\smash{{\SetFigFont{8}{9.6}{\rmdefault}{\mddefault}{\updefault}{\color[rgb]{0,0,0}$r=0$}%
}}}}
\put(3991,-2821){\makebox(0,0)[lb]{\smash{{\SetFigFont{8}{9.6}{\rmdefault}{\mddefault}{\updefault}{\color[rgb]{0,0,0}$r=0$}%
}}}}
\put(4276,-3961){\rotatebox{315.0}{\makebox(0,0)[lb]{\smash{{\SetFigFont{8}{9.6}{\rmdefault}{\mddefault}{\updefault}{\color[rgb]{0,0,0}$\mathcal{H}^-$}%
}}}}}
\put(3526,-3661){\makebox(0,0)[lb]{\smash{{\SetFigFont{8}{9.6}{\rmdefault}{\mddefault}{\updefault}{\color[rgb]{0,0,0}$\Sigma$}%
}}}}
\put(5476,-3061){\rotatebox{315.0}{\makebox(0,0)[lb]{\smash{{\SetFigFont{8}{9.6}{\rmdefault}{\mddefault}{\updefault}{\color[rgb]{0,0,0}$\overline{\mathcal{H}}^+$}%
}}}}}
\put(4351,-3361){\rotatebox{45.0}{\makebox(0,0)[lb]{\smash{{\SetFigFont{8}{9.6}{\rmdefault}{\mddefault}{\updefault}{\color[rgb]{0,0,0}$\mathcal{H}^+$}%
}}}}}
\put(4951,-3586){\makebox(0,0)[lb]{\smash{{\SetFigFont{8}{9.6}{\rmdefault}{\mddefault}{\updefault}{\color[rgb]{0,0,0}$\mathcal{D}$}%
}}}}
\put(4726,-3211){\makebox(0,0)[lb]{\smash{{\SetFigFont{8}{9.6}{\rmdefault}{\mddefault}{\updefault}{\color[rgb]{0,0,0}$\Sigma'$}%
}}}}
\put(5626,-4186){\rotatebox{45.0}{\makebox(0,0)[lb]{\smash{{\SetFigFont{8}{9.6}{\rmdefault}{\mddefault}{\updefault}{\color[rgb]{0,0,0}$\overline{\mathcal{H}}^-$}%
}}}}}
\end{picture}%

\]
to contain subsets of $\mathcal{H}^+$ 
and/or
$\overline{\mathcal{H}}^+$.

\subsubsection{Comparison with the Schwarzschild case}
The statement of Theorem~\ref{gendata}
should be compared with the results of our previous~\cite{dr1, dr3}
concerning the wave equation on a Schwarzschild exterior.

Recall that in the region $r>2M$, the Schwarzschild metric is given by
the expression  $(\ref{original})$ for $\Lambda=0$, $M>0$.
The Penrose diagramme of the closure of this region in the maximally 
extended Schwarzschild spacetime is given below:
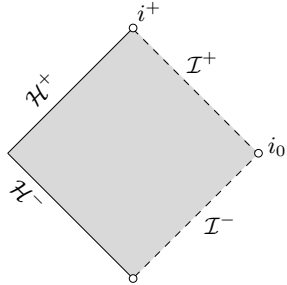
\[
\begin{picture}(0,0)%
\includegraphics{schwarzschildother2.pstex}%
\end{picture}%
\setlength{\unitlength}{2960sp}%
\begingroup\makeatletter\ifx\SetFigFont\undefined%
\gdef\SetFigFont#1#2#3#4#5{%
  \reset@font\fontsize{#1}{#2pt}%
  \fontfamily{#3}\fontseries{#4}\fontshape{#5}%
  \selectfont}%
\fi\endgroup%
\begin{picture}(2632,2325)(5134,-6100)
\put(6258,-3907){\makebox(0,0)[lb]{\smash{{\SetFigFont{9}{10.8}{\rmdefault}{\mddefault}{\updefault}{\color[rgb]{0,0,0}$i^+$}%
}}}}
\put(7351,-5011){\makebox(0,0)[lb]{\smash{{\SetFigFont{9}{10.8}{\rmdefault}{\mddefault}{\updefault}{\color[rgb]{0,0,0}$i_0$}%
}}}}
\put(5401,-4636){\rotatebox{45.0}{\makebox(0,0)[lb]{\smash{{\SetFigFont{9}{10.8}{\rmdefault}{\mddefault}{\updefault}{\color[rgb]{0,0,0}$\mathcal{H}^+$}%
}}}}}
\put(5176,-5311){\rotatebox{315.0}{\makebox(0,0)[lb]{\smash{{\SetFigFont{9}{10.8}{\rmdefault}{\mddefault}{\updefault}{\color[rgb]{0,0,0}$\mathcal{H}^-$}%
}}}}}
\put(6676,-4336){\makebox(0,0)[lb]{\smash{{\SetFigFont{9}{10.8}{\rmdefault}{\mddefault}{\updefault}{\color[rgb]{0,0,0}$\mathcal{I}^+$}%
}}}}
\put(6826,-5686){\makebox(0,0)[lb]{\smash{{\SetFigFont{9}{10.8}{\rmdefault}{\mddefault}{\updefault}{\color[rgb]{0,0,0}$\mathcal{I}^-$}%
}}}}
\end{picture}%

\]

In~\cite{dr1}, an analogue of $(\ref{fbxwris})$ is proven for all $s<6$ and spherically
symmetric initial data. Modulo an $\epsilon$,
this result is expected to be sharp, as it
not expected to be true for $s>6$, in view of heuristic arguments due to 
Price~\cite{rpr:ns}.

In~\cite{dr3}, an analogue of $(\ref{fbxwris})$ is proven for $s=2$ for arbitrary,
not necessarily spherically symmetric, initial data.

In view of the fact that solutions of the wave equation vanish on $\mathcal{I}^+$,
the results of~\cite{dr3} allow one to obtain the uniform pointwise decay rate
$|\phi|\le Cv_+^{-1}$. As a uniform decay bound in $v$,
this decay rate is in fact sharp.

\subsubsection{Second statement of the theorem}
The loss of angular derivatives in the result of Theorem~\ref{gendata}
can be more precisely quantified by decomposing $\phi$ into spherical harmonics.
Each spherical harmonic $\phi_{\ell}$ decays at least exponentially, but the bound on the
exponential rate
obtained here
\emph{decreases} inverse quadratically in the spherical harmonic number.
We have the following
\begin{theorem}
\label{second}
Let $(\mathcal{M},g)$, $\Sigma$, $\mathcal{D}$ be as in Theorem~\ref{gendata}.
Then there exists a constant $c$ depending only on $M$ and $\Lambda$,
and $C_0$ depending only on $M$, $\Lambda$, and
the geometry of $\Sigma\cap J^-(\mathcal{D})$, such that for all
$\phi_{\ell}$  
solutions of the wave equation on $\mathcal{M}$ with spherical harmonic number
$\ell$ with
\[
{\bf E}_0(\varphi_{\ell},\dot\varphi_{\ell}) 
= \|\nabla \varphi_{\ell} \|^2 + \|\dot\varphi_{\ell}\|^2 <\infty
\] 
and all achronal
hypersurfaces $\Sigma'\subset J^+(\Sigma)\cap \mathcal{D}$, 
the bound
\begin{equation}
\label{fluxbound}
\int_{\Sigma'}T_{\mu\nu}(\phi_\ell)T^\mu n^\nu \le C_0\, {\bf E}_0(\varphi_{\ell},\dot\varphi_{\ell}) 
\Big (e^{-2cv_+(\Sigma')/\ell^2}+e^{-2cu_+(\Sigma')/\ell^2}\Big)
\end{equation}
holds 
for all $\ell \ge 0$,
and, again as before, also
 with $T^\mu$ above replaced with $N^\mu$ defined in Section~\ref{SMEsec}.
 In addition, the pointwise bounds
\[
|\phi_{\ell}(u,v)|\le C_0\, {\bf E}^{\frac 12}_0(\varphi_{\ell},\dot\varphi_{\ell})  (e^{-cv_+/\ell^2}+e^{-cu_+/\ell^2})
\]
for $\ell\ge 1$, and
\[
|\phi_0(u,v)-\underline\phi|\le C_0\,{\bf E}^{\frac 12}_0(\varphi_{0},\dot\varphi_{0}) (e^{-cv_+/\ell^2}+e^{-cu_+/\ell^2}),
\]
for $\ell = 0$, hold in $J^+(\Sigma)\cap \mathcal{D}$,
where $\underline\phi$ is a constant satisfying
\[
|\underline\phi|\le \inf_{x\in \Sigma}|\phi_0(x)| + C_0\, {\bf E}^{\frac 12}_0(\varphi_0,\dot\varphi_0).
\]
\end{theorem}

The above theorem can easily be seen to imply Theorem~\ref{gendata}.

\subsection{Overview of the proof}
In this paper, we insist on a framework of proof that in principle may have 
relevance to the non-linear stability problem, that is to say, the 
problem of the dynamics of $(\ref{Eeq})$ starting from initial data close
to those induced on a Cauchy hypersurface $\Sigma$ of Schwarzschild-de Sitter.
This leads us to try to exploit compatible currents. In this section, we will describe this
general approach, and the natural relation of the currents we will define with
various geometric and analytical aspects of the problem at hand.

\subsubsection{Vector fields and compatible currents}
For quasilinear hyperbolic systems (like $(\ref{Eeq})$) in $3+1$ dimensions,
all known techniques for studying the global dynamics 
are based on $L^2$ estimates. In the Lagrangian case,
the origin of such 
estimates can be understood geometrically in terms of \emph{compatible currents}
(see Christodoulou~\cite{book2}). These are $1$-forms $J_\mu$
such that  at each point  $x\in \mathcal{M}$,
both $J_\mu$ and the divergence $K= \nabla^\mu J_\mu$
depend only on the $1$-jet of $\phi$.
An important class of these are the currents $J^V_\mu$
obtained by contracting the energy momentum tensor $T_{\mu\nu}$ with an arbitrary
vector field $V^{\mu}$.  See Section~\ref{CCsec} for a discussion in the context
of the linear wave equations $\Box_g\phi=0$ studied here.
All estimates in this paper are obtained by exploiting the integral identities
\begin{equation}
\label{introid}
\int_\mathcal{R} K = \int_{\partial \mathcal{R}} J_\mu n^\mu
\end{equation}
corresponding to
compatible currents
of the form $J^V_\mu$ and straightforward modifications thereof $J_\mu=J^V_\mu+\cdots$, 
where the vector fields $V$ are directly related
to the geometry of the problem, and the region $\mathcal{R}$ is suitably chosen.


\subsubsection{The photon sphere and the currents $J^X_\mu$}
\label{pssec}
The timelike hypersurface $r=3M$ is known as the \emph{photon sphere}.
This has the ominous property
of being spanned by null geodesics. If additional regularity is not imposed, then
it is clear by a geometric optics approximation that
solutions of the wave equation can concentrate their energy 
along such geodesics for arbitrary long times, and one can thus not achieve
a quantitative bound for the rate of decay in terms of initial energy alone. 
In particular, $(\ref{fbxwris})$ cannot
hold for $s>0$ if ${\bf E}_s$ is replaced by ${\bf E}_0$.

It is truly remarkable that this obstruction arising
from geometrical optics is captured, and quantified, 
by a current $J_\mu$ associated to a vector field $V$ of the form $f(r^*)\partial_{r^*}$
for a well-chosen function $f$.\footnote{This insight, in the case of the wave equation on
the Schwarzschild solution,
is originally due to Blue and Soffer~\cite{bs1}. See, however,~\cite{bs2}.}
The story is not entirely straightforward, however. 
The desired current is in fact not precisely of the form
$J^V_\mu$, but a modification thereof, to be denoted $J^{X,3}_\mu$,
which is associated in a well defined way to a collection of vector fields
$X_\ell=f_\ell(r^*)\partial_{r^*}$. The current is defined by summing over
currents $J^{X_\ell,3}$ which act on individual
spherical harmonics $\phi_\ell$.

The current $J^{X,3}$ yields a nonnegative $K^{X,3}$, modulo an error term 
supported near the horizons. In a first approximation, we may pretend that in fact
$K^{X,3}\ge 0$, but degenerates (in regular coordinates) near the horizon.
The identity $(\ref{introid})$ can then be used as an estimate
for its left hand side, in view of the fact that its right hand side will in fact be bounded
by the flux of $J^T$, for the Killing field $T$, which is conserved.
The role of the photon sphere will be exemplified by the degeneration at $r=3M$
of the quantity controlled by this spacetime integral.

In order to obtain decay results from the above, one would have to 
gain information about the quantity estimating the boundary terms--namely
$\int_{\partial\mathcal{R}}J^T_\mu n^\mu$, from the  control of
spacetime integral.
The difficulty for this is that the 
spacetime integral estimates one obtains degenerate at the photon sphere $r=3M$
and at the horizons.  This does not allow one to control $J^T_\mu n^\mu$ there.

The problem at the photon sphere is cured by
applying the estimate also to angular derivatives. It is here that the 
argument ``loses'' an angular derivative.
It is this loss that leads to the form
of decay proven in $(\ref{fbxwris})$.

 The problem on the horizon, on the other hand, turns out to be illusionary.
 The horizon is in fact a very favourable place for estimating the solution! For this,
 we will need to consider
 the  ``local observer'' vector fields $Y$, $\overline Y$, to be described in the next section.

\subsubsection{The red-shift effect and the currents $J^Y_\mu$ and $J^{\overline{Y}}_\mu$}
\label{rssec}
The heuristic mechanism ensuring decay near the horizons has been understood
for many years, and is known as the \emph{red shift effect}. This is typically described 
in the language of geometric optics. If two observers $A$ and $B$ cross the event horizon
at advanced times $v_A<v_B$, and $A$ sends a signal to $B$ at a certain frequency, as he ($A$)
measures it,
then the frequency at which $B$ receives it is exponentially damped in the quantity $v_B-v_A$. 

It turns out that this exponential damping property can be captured
by the integral identities $(\ref{introid})$ corresponding to the
currents $J_\mu^Y$ and $J_\mu^{\overline{Y}}$ associated to vector fields
$Y$, $\overline{Y}$, defined in Section~\ref{Ysec}.
These vector fields are
supported near the horizons $\mathcal{H}^+$, $\overline{\mathcal{H}}^+$,
respectively.
The estimates $(\ref{introid})$ corresponding to the currents $J_\mu^Y$,
$J_\mu^{\overline{Y}}$ fulfill the double role of (a) correcting for the error
region where $K^{X,3}<0$ by dominating this term near the horizon by
$K^Y+K^{\overline Y}$ and (b) controlling the spacetime integrated energy
measured by local observers near the horizon.
The choice of $Y$, $\overline{Y}$ is delicate, because there is an ``error region''
where $K^Y+K^{\overline Y}<0$, 
which must be controlled with the help of the currents of the previous section.
The use of the currents corresponding to $X$, $Y$ and $\overline{Y}$ are thus strongly coupled.

\subsubsection{Comparison with the Schwarzschild case}
To see the above arguments in context, the reader may wish to compare
with our previous~\cite{dr3}, where versions 
of the currents $J^X_\mu$, $J^Y_\mu$ are also employed. 
The relation of our arguments
with the physical mechanisms at play are in fact much clearer
in the present paper, than in~\cite{dr3}. This is due on the one hand to the absence here
of the Morawetz-type vector field (denoted $K$ in~\cite{dr3}), and, on the other hand,
to the relative simplicity here in the construction of the current $J^X_\mu$.
We give here some comments on these points. 

The Morawetz vector field employed in~\cite{dr3} is a highly unnatural quantity at the horizon from
the geometric point of view. On the other hand, in view of its weights,
it somewhat magically captures a polynomial (as opposed to the proper exponential) version of the red shift. The pointwise decay rates achieved via $K$ at the horizon are
worse than the decay rates away from the horizon, but 
sufficient if one is only interested in the behaviour of the solution
away from the horizon. (See also~\cite{bs:le}.) 
In our~\cite{dr3}, uniform decay rates
up to the horizon were indeed obtained with the help of $J^Y_\mu$. But these estimates
could be obtained \emph{a posteriori}. From the point of view
of the non-linear stability problem, this decoupling appears to be an
exceptional feature. It is in this sense that the scheme proposed in the present paper
is perhaps
more naturally connected to the geometry of general black holes.

The second point to be made here concerns the construction of $J^X_\mu$.
In~\cite{dr3}, positivity of the analogue
of what we denote here $K^{X,3}$ relied on an unmotivated recentring
and rescaling of the derivatives of the functions $f_\ell$ which obscured perhaps the fundamental
connection with the photon sphere. 
Here, this connection appears much more clear. Of course, this is at the expense
of having to bound $-K^{X,3}$ from $K^Y$  and $K^{\overline{Y}}$. 
This should in no way be thought of as a disadvantage. The red-shift effect
has a lot to offer. It should be used and not obscured.

\subsection{Discussion}
\label{discussion}
As noted above, the study of the asymptotic behaviour of
solutions to $\Box_g\phi=0$ on both Schwarzschild and
Schwarzschild-de Sitter backgrounds has a long tradition 
in the physics literature. In the Schwarzschild  case, the pioneering
heuristic study is due to Price~\cite{rpr:ns}. See also~\cite{gpp:de1}.
For the Schwarzschild-de Sitter case, there is
numerical work of Brady et al~\cite{brady},
the subsequent~\cite{brady2}, and references therein.

The above studies are based entirely upon decomposition of $\phi$ into 
spherical harmonics.
The results of these heuristics or numerics are typically presented in terms
of the asymptotic behaviour of the tail:
\begin{equation}
\label{faivetai0}
\phi_\ell(r,t) \sim  t^{-2\ell -3}, \qquad \phi_\ell(u,v)\sim v^{-2\ell -3},\qquad
r\phi_\ell(u,v)\sim u^{-2\ell-2}
\end{equation}
for Schwarzschild, where $2M<r<\infty$ is fixed in the first formula, $u\ge v$
in the second, and $v \ge 2u$ in the third,
and 
\begin{equation}
\label{faivetai}
\phi_\ell(r,t) \sim e^{-c\ell t}, \qquad \phi_\ell(u,v)\sim e^{-c\ell v},
\qquad \phi_\ell (u,v) \sim e^{-c\ell u}
\end{equation}
for Schwarzschild-de Sitter and $\ell\ge 1$, where $r_b<r<r_c$ is fixed in the first formula,
and $u\ge v$ in the second, and $v\ge u$ in the third.

At first glance, statements $(\ref{faivetai})$ may appear stronger
than what is actually proven in Theorem~\ref{second}.
As quantitative statements of decay, however, statements $(\ref{faivetai0})$
and $(\ref{faivetai})$ are in fact much weaker than what has now been mathematically
proven, here and in~\cite{dr3}.
For, rewriting, in particular, the first formula of $(\ref{faivetai})$
as 
\begin{equation}
\label{ifwewrite}
|\phi_\ell(r,t) | \le C_\ell(r) e^{-c\ell t},
\end{equation}
then there is no indication
as to what $C_\ell(r)$ depends on, indeed, if there is any bound on  $C_\ell$ provided
by some norm of
initial data, and if so, what is the behaviour as $\ell\to\infty$. This does not concern a mathematical
pathology, but is intimately connected with the physical effect caused
by the photon sphere. Indeed, a geometric optics approximation
shows easily that if $(\ref{ifwewrite})$ is to hold and
if $C_\ell$ is to depend, say, on the initial energy of the spherical harmonic, then
$C_\ell\to \infty$ as $\ell\to \infty$. \emph{It is the rate of this divergence that would 
then determine the decay rate (if any) for $\phi$.}

If one is interested in quantitative statements of
decay, a statement like
$(\ref{faivetai})$ provides no more information than the statement
\begin{equation}
\label{faivetai2}
\lim_{(u,v) \to (\infty,\infty)}\phi_\ell (u,v)=0.
\end{equation}
It is worth noting that the above statement at the level of individual spherical harmonics, 
together with the (uniform)  boundedness\footnote{The result~$(\ref{ubd})$ was shown for Schwarzschild in fundamental work of Kay and Wald~\cite{kw:lss}. 
Our~\cite{dr3} gives an alternative proof not relying on the discrete symmetries
of the maximal development. For Schwarzschild-de Sitter, the statement $(\ref{ubd})$  
of course follows
from Theorem~\ref{gendata} for any $s>1$.
We have not found another statement of this 
in the literature.}
 result
\begin{equation}
\label{ubd}
|\phi|\le C\sup|\varphi|+C \, {\bf E}_s^{\frac12}(\varphi,\dot\varphi),
\end{equation}
can indeed be used to show, for fixed $r$, the statement
\begin{equation}
\label{heuristic}
\lim_{(u,v)\to (\infty,\infty)} (\phi-\phi_0)(u,v)=0,
\end{equation}
for the total $\phi$. This can be termed the statement of
(uniform) decay without a rate.

Thus it is truly only $(\ref{heuristic})$, and not the results of~\cite{dr3} or
Theorem~\ref{gendata},
that can be said to be suggested by heuristic and numerical studies.

Results like $(\ref{heuristic})$ or even just $(\ref{ubd})$ are sometimes referred to as ``linear stability''
in the physics literature.\footnote{Sometimes, even the statement $\forall r_b<r<r_c,
\lim_{t\to\infty} \phi(r,t)=0$ is termed ``linear stability''. Such a result does not even imply
$(\ref{ubd})$. It is in fact entirely consistent
with the statement $\sup_{r\in(r_b,r_c)t\in[0,\infty)}|\phi(r,t)|=\infty$! } 
One should keep in mind, however, that
 were $(\ref{heuristic})$ the sharp
decay result, it would in fact suggest  \emph{instability} for
Schwarzschild or Schwarzschild-de Sitter once one
passes to the next order in perturbation theory.
At very least, it would exclude all known techniques for proving non-linear stability for
supercritical non-linear wave equations like $(\ref{Eeq})$. {\bf It is 
only quantitative uniform decay bounds with decay rate sufficiently fast, such that moreover
the bound depends only on a suitable norm of initial data,
which indeed can the thought of as suggestive of non-linear stability.}
One should thus be careful in associating the heuristic and numerical
tradition exemplified in~\cite{rpr:ns} with the conjecture that black holes are stable.

\subsection{Note added}
While the final version of this manuscript was being prepared, an interesting
preprint~\cite{bh} appeared addressing a special case of the problem under consideration here
with the methods of time-independent scattering theory.
The special case where $\phi$ is not supported at $\mathcal{H}^+\cap
\mathcal{H}^-$ and $\overline{\mathcal{H}}^+\cap \overline{\mathcal{H}}^-$
is considered and quantitative exponential decay
bounds are proven for $\int_{\Sigma'}T_{\mu\nu}T^\mu n^\nu$ in the coordinate
$t$, where one must restrict to
$\Sigma'= \{t\}\times [r_0,R_0]$, for $r_b<r_0$, $R_0<r_c$. The
bounds lose only an $\epsilon$ of an angular derivative,
but depend on $r_0$, $R_0$, and the initial support of $\phi$ in
an unspecified way. The work~\cite{bh} depends in an essential way on a previous detailed
analysis of S\'a Barreto and Zworski~\cite{SB-Zworski}  
concerning resonances of an associated elliptic problem.

For the special case of the data considered in~\cite{bh}, given that result, then
the estimates of the present paper, in particular, those provided by the
currents $J^Y_\mu$, $J^{\overline{Y}}_\mu$, can be applied \emph{a posteriori}
to obtain uniform (i.e.~holding up to the horizons)
exponential decay bounds.

\subsection{Acknowledgements}
The authors thank the Massachusetts Institute of Technology for hospitality
in the Spring of 2006 where this research began. M.D.~is supported
by a Clay Research Scholarship. I.R.~is supported in part by NSF grant
DMS-0702270.

\section{The Schwarzschild de-Sitter metric in coordinates}
\label{cordsec}
We refer the reader to the references~\cite{brandon, GibHawk, lake}  for detailed
discussions of the geometry of Schwarzschild-de Sitter.

\subsection{Schwarzschild coordinates $(r,t)$}
We recall that so-called Schwarzschild coordinates $(r,t)$ map
$\mathcal{D}^o$ onto $(r_b,r_c)\times(-\infty,\infty)$, 
in which the metric takes the form $(\ref{original})$. Let the choice of the $t$ coordinate
be fixed.
Here $0<r_b<r_c$ denote the two positive roots
of the equation 
\begin{equation}
\label{sxe3iswsn}
1-\frac{2M}r-\frac{\Lambda}3r^2=0.
\end{equation}

The function $r$ can be given  a geometric interpretation
\begin{equation}
\label{rdef}
r(p)=\sqrt{{\rm Area}(\hat{\pi}^{-1}(\hat\pi(p)))/4\pi},
\end{equation}
 where here $\hat{\pi}:\mathcal{M}\to\mathcal{Q}$ is the natural
projection. Thus $r$ can be defined as a smooth function on all of $\mathcal{M}$.
It is known as the \emph{area-radius function}. This function also clearly descends
to $\mathcal{Q}$.

The $(r,t)$ coordinates degenerate along the horizons $\mathcal{H}^+\cup\mathcal{H}^-$
and $\overline{\mathcal{H}}^+\cup \overline{\mathcal{H}}^-$, on which $r=r_b$,
$r=r_c$, respectively.			

It is immediate from the explicit form of the metric that the vector field $\frac\partial{\partial t}$
is Killing in $\mathcal{D}^o$. This extends to a globally defined Killing field $T$ on $(\mathcal{M},g)$,
which is null along $\mathcal{H}^+\cup\mathcal{H}^-$ and $\overline{\mathcal{H}}^+
\cup\overline{\mathcal{H}}^-$, and vanishes along
$\mathcal{H}^+\cap\mathcal{H}^-$ and $\overline{\mathcal{H}}^+
\cap\overline{\mathcal{H}}^-$.

\subsection{Regge-Wheeler coordinates $(r^*,t)$}
We now proceed to define two related coordinate systems on $\mathcal{D}^o$.
Let us denote the unique negative root of $(\ref{sxe3iswsn})$ as $r_-$, and 
let us set
\[
\kappa_{b}= \frac{d}{dr}\left(1-\frac{2M}r-\frac{\Lambda}3r^2
\right)\big|_{r=r_{b}},
\]
and similarly $\kappa_{c}$, $\kappa_-$.
We now set
\begin{eqnarray*}
r^*&\doteq&
-\frac{1}{2\kappa_{c}}\log \left|\frac{r}{r_{c}}-1
\right|+\frac{1}{2\kappa_{b}}\log
\left|\frac{r}{r_{b}}-1\right|\\
&&\hbox{}+
\frac{1}{2\kappa_{-}}\log\left|\frac{r}{r_{-}}-1\right|
-C^*
\end{eqnarray*}
where $C^*$ is a constant we may choose arbitrarily.
For convenience, let us choose
$C^*$ so that
$r^*=0$ when $r=3M$, the so-called
\emph{photon sphere}. 
We call the coordinates $(r^*,t)$ so-defined \emph{Regge-Wheeler coordinates}.

\subsection{Eddington-Finkelstein coordinates $(u,v)$}
From Regge-Wheeler coordinates $(r^*,t)$, we can define now \emph{retarded and advanced Eddington-Finkelstein coordinates}
$u$ and $v$, respectively, by
\[
t=v+u
\]
and
\[
r^*=v-u.
\]
These coordinates turn out to be null: Setting
$\mu=\frac {2M}r+\frac13\Lambda r^2,$
the metric takes the form
\[
-4(1-\mu)du dv+r^2d\sigma_{\mathbb S^2}^2.
\]

We shall move freely between the two coordinate systems $(r^*,t)$ and
$(u,v)$ in this paper.
Note that in either, region  $\mathcal{D}^o$ is covered
by $(-\infty,\infty)\times (-\infty,\infty)$.

By appropriately rescaling $u$ and $v$ to have finite range, one
can construct coordinates which are in fact regular on 
$\mathcal{H}^\pm$
and $\tilde{\mathcal{H}}^\pm$. 
By a slight abuse of language, one can parametrize the
future and past horizons in our present
$(u,v)$ coordinate systems as 
\begin{equation}
\label{conve1}
\mathcal{H}^+=\{(\infty,v)\}_{v\in[-\infty,\infty)},
\end{equation}
\begin{equation}
\label{conve2}
\mathcal{H}^-=\{(u,-\infty)\}_{u\in(-\infty,\infty]},
\end{equation}
\begin{equation}
\label{conve3}
\overline{\mathcal{H}}^+=\{(u,\infty)\}_{u\in[-\infty,\infty)},
\end{equation}
\begin{equation}
\label{conve4}
\overline{\mathcal{H}}^+=\{(-\infty,v)\}_{v\in[-\infty,\infty)}.
\end{equation}
Under these conventions, the statements of Theorem~\ref{gendata} can be applied
up to the boundary of $\mathcal{D}$.

\subsection{Useful formulae}
\label{useful}
Finally, we collect various formulas for future reference:
\[
\mu=\frac {2M}r+\frac13\Lambda r^2,
\]
\[
g_{uv}=(g^{uv})^{-1}=-2(1-\mu),
\]
\[
\partial_vr=(1-\mu),\qquad \partial_ur=-(1-\mu)
\]
\[
dt=dv+du, dr^*=dv-du,
\]
\[
T=\frac{\partial}{\partial t}=\frac12\left(\frac{\partial}{\partial v}
+\frac{\partial}{\partial
u}\right),
\]
\[
\frac{\partial}{\partial r^*}=\frac12\left(\frac{\partial}{\partial v}-\frac{\partial}{\partial
u}\right),
\]
\[
dVol_{\mathcal{M}}= 2 r^2(1-\mu)\, du\, dv\, d A_{{\mathbb S}^2},
\]
\[
dVol_{t={\rm const}}= r^2\sqrt{1-\mu}\, dr^*\, d A_{{\mathbb S}^2},
\]
\[
\Box\psi=\nabla^\alpha\nabla_\alpha\psi= -(1-\mu)^{-1}\left(
\partial_t^2\psi-r^{-2}\partial_{r^*}(r^2\partial_{r^*}\psi)\right)
+\nabb^A\nabb_A\psi.
\]
Here $\nabb$ denotes the induced covariant derivative on the group orbit
spheres.

\section{The energy momentum tensor and compatible currents}
\label{CCsec}
As discussed in the introduction, the results of this paper will rely
on $L^2$-based estimates. Such estimates arise naturally in view of
the Lagrangian structure of the wave equation. We review briefly here.

Let $\phi$ be a solution of $\Box_g\phi=0$.
In general coordinates,
the \emph{energy-momentum tensor} $T_{\alpha\beta}$ for $\phi$  is defined by the expression
\[
T_{\alpha\beta}(\phi)=\partial_\alpha\phi\, \partial_\beta\phi-\frac12g_{\alpha\beta}\,g^{\gamma\delta}
\partial_\gamma\phi\,\partial_\delta\phi.
\]
The tensor $T_{\alpha\beta}$ is symmetric and divergence-free, i.e.~we have
\begin{equation}
\label{divfree}
\nabla^\alpha T_{\alpha\beta}=0.
\end{equation}

For the null coordinate system $u,v,x^A,x^B$ we have defined, where
$x^A$, $x^B$ denote coordinates on $\mathbb S^2$, we compute the components
\[
T_{uu}=(\partial_u\phi)^2,
\]
\[
T_{vv}=(\partial_v\phi)^2,
\]
\[
T_{uv}=-\frac12g_{uv}|\nabb\phi|^2=(1-\mu)|\nabb\phi|^2.
\]
Here the notation $| \nabb \psi|^2= g^{AB}\partial_A\psi\partial_B\psi = r^{-2} |d\psi|^2_{d\sigma}$.
Note moreover
that $|\nabb\psi|^2 = r^{-2}\sum_{i=1}^3| \Omega_i\psi|^2$.

Let $V^\alpha$ denote an arbitrary vector field.
Let $\pi_V^{\alpha\beta}$ denote the deformation tensor of $V$,
i.e.,
\begin{equation}
\label{defdef}
\pi_V^{\alpha\beta}\doteq \frac12(\nabla^\alpha V^\beta+\nabla^\beta V^\alpha).
\end{equation}
In local coordinates we have the following expression:
\begin{eqnarray*}
T_{\alpha\beta}(\phi)\pi ^{\alpha\beta}_V&=&\frac{1}{4(1-\mu)}\left(
(\partial_u\phi)^2\partial_v(V_v(1-\mu)^{-1})
                                +(\partial_v\phi)^2\partial_u(V_u(1-\mu)^{-1})\right.\\
                                &&\hbox{}\left.+|\nabb\phi|^2(\partial_u
V_v+\partial_v V_u)
                                        \right)
                        -\frac1{2r}(V_u-V_v)(|\nabb\phi|^2-\phi^\alpha\phi_\alpha).
\end{eqnarray*}

Set
\begin{equation}
\label{edworizetai}
J^V_\alpha=T_{\alpha\beta}V^\alpha.
\end{equation}
The relations $(\ref{divfree})$ and $(\ref{defdef})$ give
\[
K^V\doteq\nabla^\alpha J^V_\alpha \doteq T_{\alpha\beta}(\phi)\pi^{\alpha\beta}_V,
\]
and the divergence theorem applied to an arbitrary region $\mathcal{R}$ 
gives $(\ref{introid})$.

Identity $(\ref{introid})$ is particularly useful when the vector field $V$ is Killing,
for instance the vector field $T$ defined previously. For then, $K^T=0$
and one obtains a \emph{conservation law} for
the boundary integrals. Moreover, when $\partial\mathcal{R}$ corresponds to
two homologous
timelike hypersurfaces, the
integrands $J_\mu^T(\phi)n^\mu$ on 
the right hand side of $(\ref{conid})$ when properly oriented are positive semi-definite
in the derivatives of $\phi$. 
 
Were the vector field $T$ timelike in all of $\mathcal{D}$, then
by applying 
$(\ref{introid})$ to $\phi$, $\Omega_i\phi$, etc., one could show the uniform boundedness
of all derivatives of $\phi$. Since $T$ becomes null on $\mathcal{H}^+\cup
\overline{\mathcal{H}}^+$, the integrand does not control all quantities on the horizon.
It is for this reason that even proving uniform boundedness for solutions of
$\Box_g\phi=0$ on $\mathcal{D}$ is non-trivial. (See~\cite{kw:lss}.)

For $\phi$ a solution to $\Box_g\phi=0$, the $1$-form $J^V_\mu(\phi)$ defined above has
the property that both it and its divergence $\nabla^\mu J^V_\mu$ 
depend only on the $1$-jet of $\phi$.
Following Christodoulou~\cite{book2}, we shall
call one-forms $J_\mu$ and their divergences
$K=\nabla^\mu J_\mu$ with the aforementioned property (thought of as form-valued
and scalar valued maps on the bundle of $1$-jets, respectively)
\emph{compatible currents}.

\section{Constants and cutoffs}
\subsection{The special values $r_i$, $R_i$}
\label{specialval}
In the course of this proof we shall 
require special values $r_i$, $R_i$, $i=0,1,2$, satisfying
\[
-\infty<\frac12r_0^*<2r_1^*<\frac12 r_1^*<2 r_2^*<0< 2R_2^*<\frac12 R_1^*<2 R_1^*<\frac12 R_0
<\infty.
\]

Eventually, specific choices of these
constants will be made, and these choices 
will depend only on $M$, $\Lambda$.
Constants $r_2$, $R_2$ are in fact only constrained by Lemma~\ref{vanlem}.
Constants $r_1$, $R_1$ are constrained by the necessity of satisfying 
Proposition~\ref{polukrisimo} of Section~\ref{choicesec}.

To choose $r_0$, $R_0$, on the other hand, is more subtle, as we will have to
keep track  of a certain competition of constants as $r_0$, $R_0$ vary.
We adopt, thus, the convention described in the next section.

\subsection{Dependence of constants $C$, $E$, and $\epsilon$ on $r_i$, $R_i$}
\label{depcon}
In all formulas that follow in this paper, constants which can be chosen independently
of $r_0^*$, $R_0^*$ 
shall be denoted by $C$. 
Constants $C$ will thus depend on $M$, $\Lambda$, $r_i$, $R_i$, for $i=1,2$,
and, after $r_i$, $R_i$ have been chosen, will depend only on $M$, $\Lambda$.

Constants which depend on $M$, $\Lambda$, $r_0^*$ and $R_0^*$ and tend to $0$ as
$r_0^*\to-\infty$, $R_0^*\to\infty$ will be denoted by $\epsilon$.
Finally, all other constants depending on $M$, $\Lambda$, $r_0^*$ and $R_0^*$, will
be denoted by $E$. {\bf Constants denoted by $E$ in principle diverge as $r_0^*\to\infty$,
$R_0^*\to\infty$. }

We will also use the convention $A\approx B$ whenever $C^{-1} A\le B\le C A$
with  a constant $C$ understood as above.

In view of our above conventions, note finally
 the obvious algebra of constants: $C\pm C=C$, $\epsilon C=\epsilon$, 
$CE=E$, $\epsilon E=E$, etc.

\subsection{Cutoffs}
\label{cutsec}
Associated to these special values of $r$, we will define a number of cutoff functions.
It is convenient to introduce also the notation 
\[
r^*_{-1}\doteq 4r^*_0, \qquad R^*_{-1}\doteq 4R^*_0.
\]

\subsubsection{The cutoffs $\eta_i$}
\label{cutseceta}
Let $\eta:[0,\infty)\to \mathbb R$ be a nonnegative smooth
cutoff function which is equal to $1$ in $[0,1]$ and $0$
outside $[0,2]$.

For $i=-1,0,1,2$, define 
\begin{eqnarray*}
\eta_i(r^*) &=& \eta( r^*/ r_i^* ) \qquad{\rm for\ }r^* \le 0\\
		&=& \eta (r^*/R_i^*) \qquad{\rm for\ } r^*\ge 0.
\end{eqnarray*}
Clearly $\eta_i$ has the property that $\eta_i=1$ in $[r_i^*,R_i^*]$
and $\eta_i=0$ in $(-\infty, \frac12r_{i-1}^*]\cup(\frac12 R_{i-1}^*,\infty)$.
Moreover, $\sup \eta'_i\to 0$ as $r_i^*\to-\infty$, $R_i^*\to\infty$.

\subsubsection{The cutoffs $\chi_i$ and $\overline\chi_i$}
\label{cutsecchi}
Now let $\chi:(-\infty,\infty)\to \mathbb R$, $\overline\chi:(\infty,\infty)\to\mathbb R$ be
nonnegative smooth cutoff functions such that $\chi$ is $1$ in $(-\infty,-1]$
and $0$ in $[-\frac12,\infty)$, and $\overline\chi$ is $1$ in $[1,\infty)$ 
and $0$ in $(-\infty, \frac12]$.

Now for $i=-1, 0,1$ define
\begin{eqnarray*}
\chi_i &=&\chi( -r^*/r_i^*)\\
\overline\chi_i &=&\overline\chi  (r^*/R_i^*).
\end{eqnarray*}

Clearly, $\chi_i$ has the property that $\chi_i=1$ in $(-\infty,r_i^*]$,
and $0$ in $[2 r_{i+1}^*,\infty)$. Similarly $\overline\chi_i$ has the property
that $\overline\chi_i=1$ in $[R_i^*,\infty)$, and $0$ in $(-\infty, 2R^*_{i+1}]$.

\section{The hypersurfaces $\Sigma_t$ and the region $\mathcal{R}(t_1,t_2)$}
\label{Thereg}
Let $r_1$, $R_1$ be the special values announced in Section~\ref{specialval}.
For all $t$, define 
\[
\Sigma_{t} \doteq \{t\}\times[ r_1,R_1] \cup \{(t-R_1^*)/2\}\times[(t+R^*_1)/2,\infty]
\cup[(t-r^*_1)/2,\infty]\times\{(t+r^*_1)/2\}
\]
and for any $t_2>t_1$, define
\[
\mathcal{R}(t_1,t_2)= J^+(\Sigma_{t_1})\cap J^-(\Sigma_{t_2}).
\]
The diagram below may be helpful:
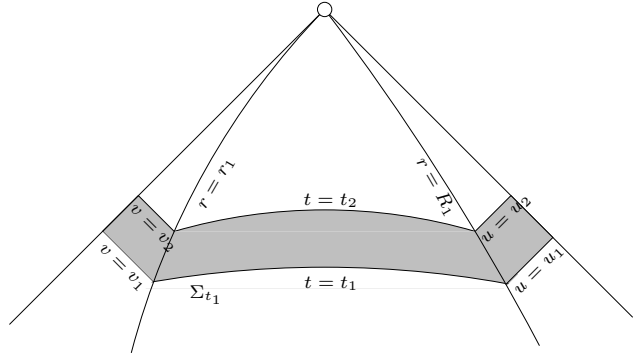
\[
\begin{picture}(0,0)%
\includegraphics{regions00.pstex}%
\end{picture}%
\setlength{\unitlength}{2368sp}%
\begingroup\makeatletter\ifx\SetFigFont\undefined%
\gdef\SetFigFont#1#2#3#4#5{%
  \reset@font\fontsize{#1}{#2pt}%
  \fontfamily{#3}\fontseries{#4}\fontshape{#5}%
  \selectfont}%
\fi\endgroup%
\begin{picture}(6549,3691)(3514,-7794)
\put(5575,-6290){\rotatebox{60.0}{\makebox(0,0)[lb]{\smash{{\SetFigFont{7}{8.4}{\rmdefault}{\mddefault}{\updefault}{\color[rgb]{0,0,0}$r=r_1$}%
}}}}}
\put(7782,-5777){\rotatebox{300.0}{\makebox(0,0)[lb]{\smash{{\SetFigFont{7}{8.4}{\rmdefault}{\mddefault}{\updefault}{\color[rgb]{0,0,0}$r=R_1$}%
}}}}}
\put(5411,-7206){\makebox(0,0)[lb]{\smash{{\SetFigFont{7}{8.4}{\rmdefault}{\mddefault}{\updefault}{\color[rgb]{0,0,0}$\Sigma_{t_1}$}%
}}}}
\put(6611,-7101){\makebox(0,0)[lb]{\smash{{\SetFigFont{7}{8.4}{\rmdefault}{\mddefault}{\updefault}{\color[rgb]{0,0,0}$t=t_1$}%
}}}}
\put(6607,-6231){\makebox(0,0)[lb]{\smash{{\SetFigFont{7}{8.4}{\rmdefault}{\mddefault}{\updefault}{\color[rgb]{0,0,0}$t=t_2$}%
}}}}
\put(4459,-6693){\rotatebox{315.0}{\makebox(0,0)[lb]{\smash{{\SetFigFont{7}{8.4}{\rmdefault}{\mddefault}{\updefault}{\color[rgb]{0,0,0}$v=v_1$}%
}}}}}
\put(4779,-6280){\rotatebox{315.0}{\makebox(0,0)[lb]{\smash{{\SetFigFont{7}{8.4}{\rmdefault}{\mddefault}{\updefault}{\color[rgb]{0,0,0}$v=v_2$}%
}}}}}
\put(8519,-6646){\rotatebox{45.0}{\makebox(0,0)[lb]{\smash{{\SetFigFont{7}{8.4}{\rmdefault}{\mddefault}{\updefault}{\color[rgb]{0,0,0}$u=u_2$}%
}}}}}
\put(8865,-7179){\rotatebox{45.0}{\makebox(0,0)[lb]{\smash{{\SetFigFont{7}{8.4}{\rmdefault}{\mddefault}{\updefault}{\color[rgb]{0,0,0}$u=u_1$}%
}}}}}
\end{picture}%

\]

We shall repeatedly apply the identity $(\ref{introid})$ in the region $\mathcal{R}(t_1,t_2)$.
We record below its explicit form:
\begin{eqnarray}
\label{conid}
0=\int_{\mathcal{R}(t_1,t_2)} K(\phi)&+& 
\int_{\Sigma_{t_2}}J_\mu(\phi) n^\mu - \int_{\Sigma_{t_1}} J_\mu(\phi) n^\mu \\
&&\hbox{}
\nonumber
+\int_{\mathcal{H}^+\cap \{v_1\le v \le v_2\} }J_\mu(\phi) n^\mu
+\int_{\overline{\mathcal{H}}^-\cap \{u_1\le u\le v_2\} }J_\mu(\phi) n^\mu.
\end{eqnarray}
Here,
$n={(1-\mu)^{-\frac 12}}T$ whenever $(u,v)$ belongs to the space-like portion of $\Sigma_t$,
and the measure of integration, call it $dm$,
is there understood to be the induced volume form.
Let us moreover define
$n\doteq \frac{\pa}{\pa u}$ and $n\doteq\frac{\pa}{\pa v}$ whenever $(u,v)$ belongs to the
$v={\rm const}$ and $u={\rm const}$
portions of $\Sigma_t$ respectively. With this choice, the measure
of integration $dm$ in the respective null segments is understood to be given by
\[
dm_{v=\text{const}}=r^2\, dA _{{\Bbb S}^2} du,\qquad dm_{u={\text{const}}}=r^2\, dA_{{\Bbb S}^2} dv.
\]

All integrals over $\Sigma_t$
 that appear in the sections that follow are understood to be with respect to the measure
$dm$ defined above.
The integral over $\mathcal{R}(t_1,t_2)$ is to be understood to be with respect to
the spacetime volume form. See Section~\ref{useful}.

\section{The main estimates}
\label{SMEsec}
In this section we will give a geometric statement of the main estimates.

\subsection{The vectorfields $N$, $\tilde{N}$ and $P$}
Recall the cutoffs functions $\eta_i$,  $\chi_1$, $\overline\chi_1$,
from Section~\ref{cutsec}.
\label{vfsec}
Define the vector fields
\begin{align}
&N\doteq \frac {\chi_1}{1-\mu} \frac{\pa}{\pa u}+\frac {\overline\chi_1}{1-\mu} \frac{\pa}{\pa v}+T,
\label{eq:N}\\
&P_{i}\doteq \eta_{i}(1-\mu)^{-1/2}\frac{\pa}{\pa r^*},\label{eq:P}\\
&\tilde{N}\doteq (r-3M)^2N\label{Ntildedef}.
\end{align}
For convenience, let us denote $P_2$ by $P$.

The above vector fields will provide the fundamental directions in which the energy
momentum tensor $T_{\mu\nu}$ is to be contracted and/or $\phi$ is to be
differentiated in the definition of the 
fundamental quantities appearing in the main estimates. 
See Section~\ref{thequansec} below.

The coordinate-dependent definitions given above notwithstanding, 
the important features of these vector fields can be understood geometrically.
All three are invariant with 
respect to the action of
$\Psi_t$, the one-parameter group of differentiable maps $\mathcal{D}\to \mathcal{D}$ generated by the Killing field $T$, $N$ is future-directed timelike on $\Sigma_t$, $\tilde{N}$ is future-directed
timelike everywhere except $r=3M$, where it vanishes quadratically, and $P$ is
supported away from the horizon and orthogonal to $T$.

\subsection{The quantities ${\bf Z}^{\tilde{N},P}$, ${\bf Z}^N$ and ${\bf Q}$}
\label{thequansec}
Let $T_{\mu\nu}$ be the energy-momentum tensor defined in Section \ref{CCsec}. Define
the quantities
\begin{align}
{\bf Z}_{\phi}^{\tilde{N},P}(t)& \doteq \int_{\Sigma_t} \left(T_{\mu\nu}(\phi)\tilde{N}^\mu n^\nu + 
({P}\phi)^2\right),\label{eq:hatZ}\\
{\bf Z}^N_\phi(t) &\doteq \int_{\Sigma_{t}} 
					T_{\mu\nu}(\phi)N^\mu n^\nu,\label{eq:justZ}\\
{\bf Q}_{\phi}(t_1,t_2)&\doteq \int_{t_1}^{t_2} {\bf Z}_{\phi}^{\tilde {N},P}(t)\, dt\label{eq:justQ}.
\end{align}

The quantity ${\bf Q}_{\phi}$ is equivalent to the spacetime integral of the
density $q(\phi)$ defined by
\[
q(\phi)\doteq\left(T_{\mu\nu}(\phi)\tilde{N}^\mu \frac{n^\nu}{1-\mu}+ 
({P}\phi)^2\right),
\]
in the sense of the formula
\begin{equation}\label{eq:ZQ}
{\bf Q}_{\phi}(t_1,t_2)\approx \int_{\mathcal{R}(t_1,t_2)}q(\phi),
\end{equation}
understood
with the conventions of Section~\ref{depcon}. We will
make use of this equivalence often in what follows.
Recall also that the spacetime integral on the right hand side of $(\ref{eq:ZQ})$ is to be understood
with respect to the  volume form. See Section~\ref{useful}.

Note that the quantity ${\bf Z}^N_\phi$ has integrand positive definite in $d\phi$. It is in fact precisely
the flux through $\Sigma_t$ of the current $J^N_\mu (\phi)$. The quantity
${\bf Z}^{\tilde{N},P}_\phi$
differs from ${\bf Z}^N_\phi$ in that control of the angular and $t$-derivatives
degenerates quadratically at $r=3M$. 
Similarly, the integrand of ${\bf Q}_\phi(t_1,t_2)$ also degenerates at
$r=3M$. This hypersurface $r=3M$ is the so-called photon sphere
discussed already in the introduction.

\subsection{Statement of the estimates}
The main estimates of this paper are contained in the following
\begin{theorem}
\label{ME}
There exists a constant $C$ depending only on $M$, $\Lambda$ such that for all $t_2>t_1$
and all sufficiently regular solutions $\phi$ of $\Box_g\phi=0$ in $\mathcal{R}(t_1,t_2)$
we have 
\begin{equation}
\label{MEstatement}
{\bf Q}_\phi(t_1,t_2)\le C\, {\bf  Z}^N_\phi(t_1),
\end{equation}
\begin{equation}
\label{MEnew}
{\bf Z}^N_{\phi}(t_2)\le C\, {\bf Z}^{\tilde{N},P}_\phi(t_2)+C(t_2-t_1)^{-1}\Big({\bf Z}^N_{\phi}(t_1)+
{\bf Q}_{\phi}(t_1,t_2)+\sum_{i=1}^3{\bf Q}_{\Omega_i\phi}(t_1,t_2)\Big),
\end{equation}
\begin{equation}
\label{MEstatementnew}
{\bf Z}^N_{\phi}(t_2)\le C\big({\bf Z}^N_{\phi}(t_1) +{\bf Q}_\phi(t_1,t_2)\big).
\end{equation}
More generally than $(\ref{MEstatementnew})$, 
if $\Sigma'\subset \mathcal{R}(t_1,t_2)$ is achronal then
\begin{equation}
\label{piogevika}
\int_{\Sigma'} T_{\mu\nu}(\phi)N^\mu n^\nu \le C\big({\bf Z}^N_\phi(t_1)+
 {\bf Q}_\phi(t_1,t_2)\big).
\end{equation}
\end{theorem}
These estimates will be used in Section~\ref{proofgendata} to prove 
Theorems~\ref{gendata} and~\ref{second}.

As described in the introduction,
the proof of Theorem~\ref{ME} 
will be accomplished in Section~\ref{proofME}
with the help of so called energy currents $J_\mu$ 
associated to the vector fields $X_{\ell}$, $Y$, $\overline{Y}$ and $\Theta$.
We turn in the next sections to the definition of these currents.

\subsection{Discussion}
Were it ${\bf Z}^{\tilde{N},P}$ on the right hand side of $(\ref{MEstatement})$, 
or alternatively,
were ${\bf Q}$ defined as the time-integral of ${\bf Z}^{N}$, 
then inequality $(\ref{MEstatement})$ would immediately lead to exponential decay in
$t$ for ${\bf Q}(t,t_*)$  (cf.~Lemma~\ref{calculus}). 

The appearance of ${\bf Q}_{\Omega_i\phi}$ on the right hand side of $(\ref{MEnew})$
signifies that the 
estimates ``lose'' an angular derivative.
At the level of any fixed spherical harmonic $\phi_\ell$, 
estimates $(\ref{MEstatement})$ and $(\ref{MEnew})$ lead
immediately to exponential decay for ${\bf Q}_{\phi_{\ell}}(t,t_*)$ and ${\bf Z}^N_{\phi}$. The nature of
the loss of angular derivative in $(\ref{MEnew})$ means 
that for the total ${\bf Q}_{\phi}(t,t_*)$ and ${\bf Z}^N_\phi$, one can only obtain
polynomial decay in $t$, where the bound on the decay rate exponent is linear
 in the angular derivatives lossed. Exponential decay for $\phi$ would
be retrieved if the ``loss in angular derivatives'' in estimate $(\ref{MEnew})$
were logarithmic. See the dependence in $\ell$ in Lemma~\ref{calculus}.

\section{The $J^{X}$ family of currents}\label{X}
We define in this section a family of currents, all loosely based on vector fields parallel
to $\frac{\partial}{\partial r^*}$. 
The role of these currents in capturing the role of the ``photon sphere'' has already been discussed in the introduction.

\subsection{Template currents $J^{V,i}_\mu$ for an arbitrary $V=f\frac{\partial}{\partial r^*}$}
\label{template}
Let $f$ be a function of $r^*$ and consider a vector field
\begin{equation}\label{eq:V}
V= f (r^*)\frac{\partial}{\partial r^*}.
\end{equation}
Define the currents
\begin{eqnarray*}
J_\mu ^{V,0}(\phi)&=& T_{\mu\nu}(\phi){V}^\nu,\\
J_\mu^{V, 1}(\phi) &=& T_{\mu\nu}(\phi){V}^\nu+\frac 14\left(f'+2\frac{1-\mu}r f\right) \pa_\mu (\phi)^2-
\frac 14 \pa_\mu \left(f'+2\frac{1-\mu}r f\right) \phi^2,\\
J_\mu^{V,2}(\phi) &=&T_{\mu\nu}(\phi)V^\nu +\frac 14\left(f'+2\frac{1-\mu}r f\right) \pa_\mu (\phi)^2-
\frac 14 \pa_\mu \left(f'+2\frac{1-\mu}r f\right) \phi^2\\
&&\hbox{}-\frac 12 \frac {f'}{r f} V_\mu \phi^2,\\
J_\mu^{V,3}(\phi)&=&T_{\mu\nu}(\phi)V^\nu +\frac 14\left(f'+2\frac{1-\mu}r f\right) \pa_\mu (\phi)^2-
\frac 14 \pa_\mu \left(f'+2\frac{1-\mu}r f\right) \phi^2\\
&&\hbox{}-\frac 12 \frac {f'}{r f} V_\mu \phi^2-\frac 12 \frac{1-3M/r}{r}
f\,\phi\,\nabb_\mu\phi,\\
J_\mu^{V, 4}(\phi) &=& T_{\mu\nu}(\phi){V}^\nu+\frac 14f' \pa_\mu (\phi)^2-
\frac 14 \pa_\mu f' \phi^2,
\end{eqnarray*}
and the divergences
\[
K^{V,i}= \nabla^\mu J_{\mu}^{V,i}.
\]

Note the identities
\begin{equation}
\label{noteid1}
\frac{\mu'}{2(1-\mu)}+
		\frac{1-\mu}{r}=\frac{r-3M}{r^2},
\end{equation}
\[
\frac{1}{2r}\left(\frac{\mu''}{1-\mu}-\frac{\mu'}r\right)
=
\frac{M}{r^4}\left(3-\frac{8M}r\right)+\frac{M\Lambda}{3r^2}-\frac{2\Lambda^2 r}9.
\]

We compute
\begin{align}
K^{V,0}(\phi) &=
\frac{f'(\partial_{r^*}\phi)^2}{1-\mu}
+|\nabb\phi|^2\left(\frac{\mu'}{2(1-\mu)}+
		\frac{1-\mu}{r}\right)f\notag\\
&-\frac14\left(2f'+4\frac{1-\mu}r f \right)\phi^\alpha
\phi_\alpha,\label{eq:KV0}
\end{align}
\begin{eqnarray*}
K^{V,1}(\phi) &=&
	\frac{f'}{1-\mu}(\partial_{r^*}\phi)^2
		+|\nabb\phi|^2\left(\frac{\mu'}{2(1-\mu)}+
		\frac{1-\mu}{r}\right)f\\
		&&\hbox{}-\frac 14 \left (\Box \left(f'+2\frac{1-\mu}rf\right) 
		\right) \phi^2\\
		&=& 	\frac{f'}{1-\mu}(\partial_{r^*}\phi)^2
		+|\nabb\phi|^2\left(\frac{\mu'}{2(1-\mu)}+
		\frac{1-\mu}{r}\right)f\\
		&&\hbox{}-\frac 14 \left (
		 \frac{1}{1-\mu}f'''+\frac 4rf''-\frac{4\mu'}{r(1-\mu)}f'
	+\frac{2}{(1-\mu)r}\left(\frac{\mu'(1-\mu)}r-\mu''\right)f
		\right)\phi^2,
\end{eqnarray*}
\begin{eqnarray*}
K^{V,2}(\phi) &=&
\frac{f'}{(1-\mu) r^2}\left(\partial_{r^*} (r\phi)\right)^2+
\frac{1-3M/r}{r} \, f\, |\nabb\phi|^2-\frac14\frac1{1-\mu}{f'''}\phi^2\\
&&\hbox{}+ f\left(\frac{M}{r^4}\left(3-\frac{8M}r\right)
+\frac{M\Lambda}{3r^2}-\frac{2\Lambda^2r}9\right) \phi^2,
\end{eqnarray*}
\begin{eqnarray*}
K^{V,3}(\phi) &=&
\frac{f'}{(1-\mu) r^2}\left(\partial_{r^*} (r\phi)\right)^2
-\frac{1-3M/r}{r} \, f\, \phi\Trb\phi-\frac14\frac1{1-\mu}{f'''}\phi^2\\
&&\hbox{}+ f\left(\frac{M}{r^4}\left(3-\frac{8M}r\right)
+\frac{M\Lambda}{3r^2}-\frac{2\Lambda^2r}9\right) \phi^2,
\end{eqnarray*}
\begin{eqnarray*}
K^{V,4}(\phi) &=&
	\frac{f'}{1-\mu}(\partial_{r^*}\phi)^2
		+|\nabb\phi|^2\left(\frac{\mu'}{2(1-\mu)}+
		\frac{1-\mu}{r}\right)f\\
		&&\hbox{}-\frac{1-\mu}r f\phi^\alpha
\phi_\alpha-\frac 14 \left (\Box f'\right) 
		\phi^2\\
		&=& 	\frac{f'}{1-\mu}(\partial_{r^*}\phi)^2
		+|\nabb\phi|^2\left(\frac{\mu'}{2(1-\mu)}+
		\frac{1-\mu}{r}\right)f\\
		&&\hbox{}-\frac{1-\mu}r f\phi^\alpha
\phi_\alpha-\frac 14 \left (
		 \frac{1}{1-\mu}f'''+\frac 2rf''		\right)\phi^2.
\end{eqnarray*}

The expression $J^{V,3}_\mu$ is not a  compatible current in the sense of
Section~\ref{CCsec}, 
since $K^{V,3}$ depends
on the $2$-jet of $\phi$, but it can be treated as such when
restricted to eigenfunctions of $\Trb$.

\subsection{Discussion}

The relation of the photon sphere to currents based on
vector fields $V$ of the form $f(r^*)\partial_{r^*}$ is most clear upon
examining the modified current $J^{V,1}_\mu$ and noting the coefficient
of $|\nabb\phi|^2$ in $K^{V,1}$ vanishes precisely at $r=3M$ in view of $(\ref{noteid1})$.
This indicates that if one is to have say $K^{V,1}\ge 0$, the function $f$ must change sign at $r=3M$, 
and the control of the angular derivatives must degenerate at least quadratically.

The task of choosing a suitable $f$ is simplified by passing to the
further modified current $J^{V,2}$
which effectively ``borrows'' positivity from the $\partial_{r^*}\phi$ term. 
Finally, one can take advantage of the further flexibility provided by chosing separately
$V$ for each spherical harmonic, where now,
after passing to the current $J^{V,3}_\mu$, 
the $0$'th order terms are united with the
angular derivative terms.  In Section~\ref{tvfX},
we shall construct
a current 
\[
J^{X,3}_\mu(\phi)\doteq J^{X_0}_\mu(\phi_0)+ \sum_{\ell} J^{X_{\ell},3}_\mu(\phi_{\ell}).
\]
We do not in fact ensure that $K^{X,3}\ge 0$ everywhere, but rather, only in the region
$r_0\le r\le R_0$. The region near the horizon will be handled with the help of
the spacetime integral terms controlled by the 
currents $J^Y$, $J^{\overline{Y}}$ to be discussed in the next section.

Once an initial positive definite spacetime integral (albeit modulo an error) is constructed,
other quantities can be controlled with the help of auxiliary currents. These 
are defined in Section~\ref{auxcursec}. It is there where we also use the current template $J^{V,4}$.

\subsection{The vector fields $X_\ell$}
\label{tvfX}
\subsubsection{The case $\ell=0$}
\label{casel0}
Define 
\[
f_0=-r^{-2},
\]
and set $X_0=f_0\partial_{r^*}$.
Let $\phi_0$ denote the $0$'th spherical harmonic of a solution $\Box_g\phi=0$
of the wave equation.
Consider $J_\mu^{X_0,0}(\phi_0)$, $K^{X_0,0}(\phi_0)$ as defined above.
We compute
\begin{equation}\label{eq:f0}
K^{X_0,0}(\phi_0) = 2r^{-3}\frac{(\partial_{r^*}\phi_0)^2}{1-\mu}.
\end{equation}
Note that
\[
K^{X_0,0}(\phi_0)\ge 0.
\]

\subsubsection{The case $\ell\ge 1$}
\label{caselg1}
Consider for now the expression $K^{V,3}$ for a general vector field $V$ of the form
\eqref{eq:V}, applied to a spherical harmonic $\phi_\ell$ with 
spherical harmonic number $\ell\ge 1$. 
We have
\begin{eqnarray}
\label{kv2}
K^{V,3}(\phi_\ell)&=&
\frac{f'}{(1-\mu) r^2}\left(\partial_{r^*} (r\phi_\ell)\right)^2-
\frac14\frac1{1-\mu}{f'''}\phi_\ell^2\\
\nonumber
&&\hbox{}+ f\left(\ell(\ell+1) \frac {1-3M/r}{r^3}+\frac{M}{r^4}\left(3-\frac{8M}r\right)
+\frac{M\Lambda}{3r^2}-\frac{2\Lambda^2r}9\right) \phi_\ell^2.
\end{eqnarray}
Define
\[
h_\ell(r)= 
\ell(\ell+1) \frac {1-3M/r}{r^3}+\frac{M}{r^4}\left(3-\frac{8M}r\right)
+\frac{M\Lambda}{3r^2}-\frac{2\Lambda^2r}9.
\]
The following lemma is proven in Appendix~\ref{vanapp}
\begin{lemma}
\label{vanlem}
For all $\ell\ge 1$, there exists a unique zero $r_{h_\ell}$ of the function $h_\ell(r)$ in $[r_b,r_c]$,
and there exist constants $r_2^*<0$, $R_2^*>0$, depending only on $M$ and $\Lambda$, 
such that 
\[
r_2<r_{h_\ell}<R_2.
\]
Moreover, 
$\lim_{\ell\to \infty} r_{h_\ell}\to 3M$. 
\end{lemma}
\emph{Let $r_2^*$,  $R_2^*$ now be fixed, chosen according to the above lemma.}

Were the middle term on the right hand side of $(\ref{kv2})$ absent,
then we would have $K^{V,3}(\phi_\ell)\ge 0$ for any $f$ such that $f'\ge 0$,
$f(r_{h_\ell})=0$.

The middle term of $(\ref{kv2})$ vanishes if $f'''=0$, but, with the requirement that $f(r_{h_\ell})=0$, 
in this case the function $f$
cannot be bounded. By suitably cutting off a function linear in $r^*$, we can ensure 
that $K^{V,3}(\phi_\ell)\ge 0$ in $r_0\le r\le R_0$.

Let $\eta_0$ be the cutoff of Section~\ref{cutsecchi}.
%
%
%
Define 
\[
f_{\ell}(r^*) = \int _{r_{h_{\ell}}^*}^{r^*} \eta_0\,  dr^*,
\]
and
the vector field
$X_{\ell}$ by
\begin{equation}
\label{thisform}
X_{\ell}=\frac12 
f_{\ell}\frac{\partial}{\partial v}-\frac12f_{\ell}\frac{\partial}{\partial u}
=f_{\ell}\frac{\partial}{\partial r^*}.
\end{equation}

We have
\begin{equation}
\label{initposit}
K^{X_{\ell}, 3}(\phi_\ell) \ge 0
\end{equation}
in the region
$r^*_0\le r \le R^*_0$. Moreover, in this region we have in fact,
\begin{equation}\label{eq:X3}
\frac{\left(\pa_{r^*}(r\phi_\ell)\right)^2}{1-\mu}+(r-r_{h_\ell})^2 \phi_\ell^2
\le 
C K^{X_{\ell}, 3}(\phi_\ell).
\end{equation}
(Recall our conventions for constants $C$ from Section~\ref{depcon}.)

\subsubsection{The currents $J^{X,i}$}\label{Aux}
Finally, for $i=1,2,3$ define the ``total'' currents
\[
J_\mu^{X,i}(\phi)= J^{X_0,0}_\mu+ \sum_{\ell\ge 1} J_\mu^{X_\ell,i}(\phi_\ell),
\]
and their divergences
\[
K^{X_\ell,i}(\phi_\ell)= \nabla^\mu J_{\mu}^{X_\ell,i}(\phi_\ell),
\]
\[
K^{X,i}(\phi)=K^{X_0,0}(\phi_0)+ \sum_{\ell\ge 1} K^{X_\ell,i}(\phi_\ell) = 
\nabla^\mu J_\mu^{X,i}(\phi).
\]

\subsubsection{Controlling the error}
\label{errcontrol}

Besides obtaining nonnegativity for $K^{X,3}(\phi)$
in the region $r_0\le r\le R_0$, we need to understand
the error in the region $r\le r_0$ and $r\ge R_0$.
It turns out that this error can be controlled by\footnote{Recall the conventions
regarding constants $\epsilon$ in Section~\ref{depcon}.} 
$\epsilon\hat{q}(\phi)$, where $\hat{q}(\phi)$ is a slightly
stronger quantity than the energy density $q(\phi)$. 

Define the quantity
\begin{eqnarray}
\label{hatqpw}
\hat{q}(\phi)&\doteq& \left (T_{\mu\nu}(\phi)\tilde{N}^\mu \frac{n^\nu}{1-\mu}+
\frac{\eta_{-1}(\chi_1+\overline\chi_1)}{1-\mu}\,|r^*|^{-\delta-1} |\nabb\phi|^2+ \eta_1 (P \phi)^2
\right).
\end{eqnarray}
Here $\eta_{-1}$, $\chi_{1}$, $\overline{\chi}_{1}$ 
are the cut-off functions defined in Section \ref{cutsec}.
Note that 
\begin{align}\label{eq:qphi}
\hat q(\phi)\approx \chi_1 \frac {(\pa_u\phi)^2}{(1-\mu)^2} &+  
\overline\chi_1 \frac {(\pa_v\phi)^2}{(1-\mu)^2}+\frac{\eta_{-1}(\chi_1+\overline\chi_1)}{1-\mu}
\,|r^*|^{-\delta-1} |\nabb\phi|^2\\&+(r-3M)^2\left ( (\pa_t\phi)^2 + |\nabb\phi|^2\right)
+(\pa_{r^*}\phi)^2\notag.
\end{align}
Compare with  
\begin{align}\label{eq:justqphi}
q(\phi)\approx \chi_1 \frac {(\pa_u\phi)^2}{(1-\mu)^2} &+  
\overline\chi_1 \frac {(\pa_v\phi)^2}{(1-\mu)^2}+(\pa_{r^*}\phi)^2\\&+(r-3M)^2\left ( (\pa_t\phi)^2 + |\nabb\phi|^2\right)
\notag,
\end{align}
where $q(\phi)$ is 
the density of the main quantity ${\bf Q}_\phi$ defined in \eqref{eq:justQ}.

The inequality replacing $(\ref{initposit})$ which holds globally is given by
\begin{lemma}
\label{controllem}
The inequality 
\begin{equation}
\label{itsufhere}
\int_{\mathbb S^2} K^{X,3}(\phi)\ge -\epsilon \int_{\mathbb S^2} \hat q({\phi})
\end{equation}
holds on all spheres of symmetry.
\end{lemma}
%
\begin{proof}
It suffices to consider the regions $r\le r_0$ and $r\ge R_0$.
Relation \eqref{eq:qphi} implies
$$
C \hat q(\phi) \ge 
 \left(1+\frac {\eta_{-1}(\chi_1+\overline\chi_1)}{1-\mu}\right) \, |r^*|^{-\delta-1}\,|\nabb\phi|^2,
$$
in these two regions,
while 
\begin{align*}
\int_{\mathbb S^2} K^{X,3}(\phi)&\ge \sum_{\ell\ge 1}
\int_{\mathbb S^2} K^{X_\ell,3}(\phi_\ell)\\ &\ge -\sum_{\ell\ge 1}
\int_{\mathbb S^2} \frac{f_\ell'''}{1-\mu}\phi_\ell^2.
\end{align*}
The function $f_\ell'''$ is supported in the region 
$[2r_0^*,r_0^*]\cup [R_0^*, 2R_0^*]$ and obeys the pointwise bound 
\begin{equation}
\label{byourapp}
|f_\ell'''(r^*)|\le C |r^*|^{-2},
\end{equation}
under our conventions for the constants denoted $C$.
Note that  $\eta_{-1}(\chi_1+\overline\chi_1)=1$ in the support of $f'''$.
Therefore,
$$
\int_{\mathbb S^2}K^{X,3}(\phi)\ge -C
\int_{\mathbb S^2}\frac 1{1-\mu} |r^*|^{-2} \, (\phi-\phi_0)^2
$$
holds on all spheres of symmetry.
We obtain $(\ref{itsufhere})$
 with 
$$
\epsilon=C\left(|r_0^*|^{-1+\delta}+|R_0^*|^{-1+\delta}\right).
$$
\end{proof}

\subsection{Auxiliary currents}
\label{auxcursec}
We will also need several ``auxiliary'' 
currents. 

\subsubsection{Auxiliary positive definite pointwise quantities}
\label{apdpq}
Let us first define, however, certain auxiliary positive definite quantities.
The auxiliary currents $K$ will be seen to bound these quantities
when integrated on spheres of symmetry in the region $r_0\le r\le R_0$.

Define
\begin{eqnarray*}
q_i(\phi) &\doteq&\eta_i q(\phi) = \eta_i (T_{\mu\nu}(\phi)\tilde{N}^\mu n^\nu(1-\mu)^{-1} + (P_2 \phi)^2)\\
 q^a_i(\phi) &\doteq& \eta_i (P_{-1}\phi)^2,\\
q^{a'}_i(\phi) &\doteq& \eta_i (\phi-\phi_0)^2,\\
q^b_i(\phi)&\doteq& \eta_i (r-3M)^2|\nabb\phi|^2,\\
q^d_i(\phi) &\doteq& \eta_i (r-3M)^2(T \phi)^2.
\end{eqnarray*}
Here $\eta_i$
are the cut-off functions defined in Section \ref{cutsec}.
Note that 
\begin{equation}
 \label{eq:qphi0}
q_1(\phi)\approx C({q}^a_1(\phi)+{q}^b_1(\phi)+{q}^d_1(\phi)),
\end{equation}
\[
q^x_{i+1}\le q^x_{i}
\]
for $x=\emptyset, a, a', b, d$, 
and that
\[
r^{-2} \ell(\ell+1)\int_{\mathbb S^2} (r-3M)^2q_i^{a'}(\phi_\ell) r^2 dA_{\mathbb S^2}
= \int_{\mathbb S^2} q_i^b(\phi_\ell)r^2 dA_{\mathbb S}^2
\]
for all spheres of symmetry.

The currents to be described in what follows are motivated by the problem of
bounding the positive definite quanitites whose latin superscripts they share.

\subsubsection{The current $J^{{X^a},3}$}

Define
\[
{f}^a_{\ell}(r^*)\doteq -\frac16\eta_2(r^*) (r^*-r_{h_\ell})^3,
\]
where $\eta_2$ is as defined in Section~\ref{cutsecchi}.

Note that $({f}^a_{\ell})'''=-1$ on $[r_2,R_2]$,
${f}^a_\ell (r_{h_\ell})=0$, 
and ${f}^a_{\ell}=0$ for $r^*\le 2r^*_2$, $r^*\ge 2 R^*_2$.

Set ${X}^a_\ell = f^a_\ell\partial_{r^*}$ 
and define
\[
J^{{X^a},3}_\mu(\phi) = \sum_{\ell\ge 1} J^{{X}^a_{\ell},3}_\mu(\phi_{\ell}),
\]
\[
K^{{X}^a,3}(\phi)= \sum_{\ell\ge 1} K^{{X}^a_{\ell},3}(\phi_\ell)=
\nabla^\mu J^{{X}^a,3}_\mu(\phi).
\]

The current $J^{X^a,3}(\phi)$ will allow us to bound the spacetime integrals of the quantities
${q}^a_0(\phi)$ and $q^{a'}_0(\phi)$. See Lemma~\ref{q1lem}. 
At this point, we can
see from \eqref{kv2} and \eqref{eq:X3} 
that for each $\ell$, the pointwise bound
\begin{equation}\label{eq:Xa}
{q}^a_1(\phi_\ell)+q^{a'}_0 (\phi_\ell) +\ell(\ell+1)(r-r_{h_\ell})^2q^{a'}_0(\phi_\ell)
\le
C K^{{X}_\ell,3}(\phi_\ell)+C K^{{X}^a_\ell,3}(\phi_\ell)
\end{equation}
holds
in $r_0\le r\le R_0$.

\subsubsection{The currents $J^{{X^b},2}$, $J^{X^b,0}$}

Define 
\[
f^b(r^*) \doteq \eta_2(r^*) (r-3M)
\]
Set ${X}^b=f^b \frac{\partial}{\partial r^*}$
 and define as before the currents
$J^{{X}^b,2}(\phi-\phi_0)$, $J^{{X}^b,0}(\phi)$ and $K^{{X}^b,2}(\phi-\phi_0)$, $K^{{X}^b,0}(\phi)$.

The current $J^{{X}^b,2}(\phi-\phi_0) $ will allow us to bound--in addition to the previous--the
spacetime integral of the quantity $q^b_1(\phi)$. See Lemma~\ref{q3lem}. 
At this point, we can deduce from  \eqref{eq:f0}, \eqref{kv2} and \eqref{eq:Xa} the
bound
\begin{align}\label{eq:Xb}
\int_{{\Bbb S}^2} &\left ( q_0^a(\phi) +q_0^{a'}(\phi)+ q_0^b(\phi) \right) 
r^2 dA_{{\Bbb S}^2}\\
&\le
\int_{{\Bbb S}^2} \left(C K^{X,3}(\phi) +C K^{X^a,3}(\phi) + C K^{X^b,2}(\phi-\phi_0)\right)
r^2 dA_{{\Bbb S}^2}\notag
\end{align}
on each sphere of symmetry with $r^*\in [r_0^*, R_0^*]$.

The current $J^{{X}^b,0}(\phi)$, in conjunction also with $J^{X^c,0}(\phi_0)$ and
$J^{X^d,4}(\phi-\phi_0)$ to be defined below,
will be used to help bound the quantity ${\bf Z}^N_\phi$
in  Proposition~\ref{metaaux}. For now, note that
 from  \eqref{eq:KV0},
\eqref{eq:f0}, \eqref{kv2} and \eqref{eq:Xb}, the estimate
\begin{align}\label{eq:Xb0}
 \int_{{\Bbb S}^2} &\Big(q^a_0(\phi)+q_0^{a'}(\phi)+ q^b_0(\phi)
\\
&\qquad-C\big(\eta_2(1-\mu)(2(r-3M)/r +1)+\eta_2'(r-3M)\big)\, \pa^\mu \phi\,\pa_\mu\phi\Big) r^2 dA_{{\Bbb S}^2},\notag\\
&\le
\int_{{\Bbb S}^2} \left(C K^{X,3}(\phi) +C K^{X^a,3}(\phi) + C K^{X^b,2}(\phi-\phi_0) +C K^{X^b,0}(\phi)\right)
r^2 dA_{{\Bbb S}^2}
\notag
\end{align}
holds on each sphere of symmetry with $r^*\in [r_0^*,R_0^*]$.
Note that the left hand side of \eqref{eq:Xb0}
\emph{a priori} does not have a sign, in view of the term containing $\partial^\mu\phi\partial_\mu
\phi$. After defining $J^{X^c,0}(\phi_0)$ and
$J^{X^d,4}(\phi-\phi_0)$  below, we shall be able to 
improve~\eqref{eq:Xb0} with~\eqref{eq:Xbd02}.
\vskip 1pc

\subsubsection{The current $J^{{X^c},0}$}
Define 
\[
f^c(r^*)\doteq r^2.
\]
Set $X^c=f^c\frac{\partial}{\partial r^*}$,
and define as before $J^{{X}^c,0}$, $K^{X^c,0}$.

We have that
\begin{equation}
\label{giatomndeviko}
K^{X^c,0}(\phi_0)= 2r(\partial_t\phi_0)^2.
\end{equation}

\subsubsection{The current $J^{{X^d},4}$}
Let us then finally define
\[
f^d(r^*)\doteq \eta_1(r^*) (r-3M)^3.
\]
Set ${X}^d=f^d \frac{\partial}{\partial r^*}$
 and define as before the currents
$J^{{X}^d,4}$ and $K^{{X}^d,4}$. 

The currents $J^{X^c,0}(\phi_0)$ and $J^{{X}^d,4}(\phi-\phi_0)$ 
will allow us to bound the spacetime integral of the quantity ${q}^d_1(\phi)$.
See Lemma~\ref{q4lem}. For now, it follows from \eqref{eq:KV0}, \eqref{eq:Xb} and
\eqref{giatomndeviko} that on each sphere of symmetry with
$r^*\in [r_0^*, R_0^*]$, we have
\begin{align}\label{eq:Xd}
 \int_{{\Bbb S}^2} &\left (q_0^a(\phi) + q_1^{a'}(\phi)+
q_0^b(\phi)+ q_1^d(\phi)  \right) r^2 dA_{{\Bbb S}^2}\\
&\le
\int_{{\Bbb S}^2} \left (C K^{X,3}(\phi) +C K^{X^a,3}(\phi) + C K^{X^b,2}(\phi-\phi_0)+
CK^{X^c,0}(\phi_0)\right.\notag \\
&\qquad\left.+
C K^{X^d,4}(\phi-\phi_0)\right)
r^2 dA_{{\Bbb S}^2}.\notag
\end{align}

The currents $J^{{X}^b,0}(\phi)$, $J^{X^c,0}(\phi_0)$ and $J^{X^d,4}(\phi-\phi_0)$ together allow
us to estimate, on each sphere of symmetry for $r^*\in[r_0^*,R_0^*]$, the quantity
\begin{align}\label{eq:Xbd02}
\int_{{\Bbb S}^2} &C \eta_2(2(r-3M)/r +1) ((\partial_t\phi)^2-|\nabb\phi|^2)\, r^2 dA_{{\Bbb S}^2}
\\
&\le
\int_{{\Bbb S}^2} \left (C K^{X,3}(\phi) +C K^{X^a,3}(\phi) + C K^{X^b,2}(\phi-\phi_0)+
CK^{X^b,0}(\phi)\right.\notag\\
&\qquad\left.+CK^{X^c,0}(\phi_0)+ C K^{X^d,0}(\phi-\phi_0)\right)
r^2 dA_{{\Bbb S}^2}.
\notag
\end{align}
Again, the left hand side of $(\ref{eq:Xbd02})$ does not have a sign.

\section{The currents $J^Y$ and $J^{\overline{Y}}$}
\label{Ysec}
In this section, we define the currents $J_\mu^Y$ and $J^{\overline{Y}}_\mu$, associated
to two vector fields $Y$ and $\overline{Y}$, supported near the black hole and cosmological
horizons, respectively. The role of these currents in capturing the ``red-shift'' effect has
been discussed in the introduction.

\subsection{The vector fields $Y$ and $\overline{Y}$}
Let $0<\delta<1 $ be a small number,
and define functions $\alpha(r^*)$, $\beta(r^*)$ as follows.
Recall the cutoff functions $\chi_1$, $\eta_{-1}$ from Section~\ref{cutsecchi},
and define
\begin{equation}
\label{alphadef}
\alpha(r^*)=\chi_1(r^*)(2-\mu+ \eta_{-1}(r^*) |r^*|^{-\delta}),
\end{equation}
\begin{equation}
\label{betadef}
\beta(r^*)=2r_b ^{-2}(2M/{r_b}^2 -2\Lambda r_b/3)^{-1} \chi_1(r^*) (1-\mu+ \eta_{-1}(r^*) |r^*|^{-\delta}).
\end{equation}

We have $\alpha\ge0$, $\beta\ge0$. To see the latter, 
note that 
$$
\frac d{dr} (1-\mu)\Big|_{r=r_b}=\left(\frac {2M}{r^2} -\frac 23 \Lambda r\right)\Big|_{r=r_b}>0
$$
since $r_-<r_b<r_c$ are three non-degenerate roots of $(1-\mu)$, which is non-negative on 
$[r_b,r_c]$. The above expression approaches $0$ as $M$ and $\Lambda$ tend
to the extremal values.

Similarly, define functions 
 $\overline{\alpha}$, $\overline{\beta}$ as follows.
Recall the cutoff $\overline\chi_1$ from Section~\ref{cutsecchi}, and
define
\begin{equation}
\label{alphadef2}
\overline\alpha(r^*)=\overline\chi_1(r^*)(2-\mu+\eta_{-1} |r^*|^{-\delta}),
\end{equation}
\begin{equation}
\label{betadef2}
\overline\beta(r^*)=2r_c ^{-2}(2\Lambda r_c/3-2M/{r_c}^2)^{-1} \overline
\chi_{1}(r^*) (1-\mu+\eta_{-1} |r^*|^{-\delta}).
\end{equation}
Again, note that $\overline\alpha\ge 0$, $\overline\beta\ge 0$.

Set $Y$ to be the vector
field
\[
Y=\frac {\alpha(r^*)}{1-\mu}\frac{\partial}{\partial u} +\beta(r^*)\frac{\partial}{\partial v}
\]
and $\overline{Y}$ to be
\[
\overline{Y} =\frac {\overline{\alpha}(r^*)}{1-\mu}\frac{\partial}{\partial v} +
\overline{\beta}(r^*)\frac{\partial}{\partial u}.
\]

The application of cutoff $\eta_{-1}$  
is  in fact not essential for the arguments of the paper. The cutoff simply
ensures that $Y$ and $\overline{Y}$ have smooth extensions beyond $\mathcal{H}^+$
and $\overline{\mathcal{H}}^+$.

\subsection{Definition of the currents}
Define the currents 
\[
J_\mu^Y(\phi) = T_{\mu\nu}Y^\nu (\phi),
\]
\[
J_\mu^{\overline{Y}}(\phi)=T_{\mu\nu} \overline{Y}^\nu (\phi),
\]
and set
\[
K^Y(\phi)=\nabla^\mu J^Y_\mu(\phi),
\]
\[
K^{\overline{Y}}(\phi)=\nabla^\mu J^{\overline{Y}}_\mu(\phi).
\]
We have
\begin{eqnarray}
\label{toexoumeauto}
\nonumber
K^Y(\phi)&=&T_{\mu\nu}(\phi)\pi_Y^{\mu\nu}\\
	&=&\frac{(\partial_u\phi)^2}{2(1-\mu)^2}
				\left(\alpha \left(\frac{2M}{r^2}-\frac{2
					\Lambda r}3\right) - \alpha'\right)
					+\frac{(\partial_v\phi)^2}{2(1-\mu)}\beta'\\
					&&\hbox{}
\nonumber
		+\frac 12|\nabb\phi|^2\left (\frac {\alpha'}{1-\mu} -
		 \frac {\left(\beta(1-\mu)\right)'}{1-\mu}\right)
				-\frac{1}{r}\left(\frac{\alpha}{1-\mu} -
				\beta\right)\pa_u\phi\, \pa_v\phi
\end{eqnarray}
and similarly 
\begin{eqnarray}
\label{toexoumeauto2}
\nonumber
K^{\overline Y}(\phi)&=&T_{\mu\nu}(\phi)\pi_{\overline Y}^{\mu\nu}\\
		&=&\frac{(\partial_v\phi)^2}{2(1-\mu)^2}
				\left(-\overline\alpha \left(\frac{2M}{r^2}-\frac{2
					\Lambda r}3\right) + \overline\alpha'\right)
					-\frac{(\partial_u\phi)^2}{2(1-\mu)}\overline\beta'\\
					&&\hbox{}
		-\frac 12|\nabb\phi|^2\left ({\overline\alpha'} -
		  {\left(\overline\beta(1-\mu)\right)'}\right)
				+\frac{1}{r}\left({\overline\alpha} -
				\overline\beta(1-\mu)\right)\pa_u\phi\, \pa_v\phi. 
\nonumber
\end{eqnarray}

\subsection{Discussion and the choice of $r_1$, $R_1$}
\label{choicesec}
In this section we exhibit the ``red-shift'' property of the currents $J^Y$, $J^{\overline{Y}}$. This
is essentially contained in Proposition~\ref{polukrisimo} below. In the context of proving this
proposition, we will choose the constants $r_1$, $R_1$.

Note that the polynomial powers in the definitions $(\ref{alphadef})$ $(\ref{betadef})$ would be unnecessary
had they not been present in the definition $(\ref{hatqpw})$. (Their presence
in $(\ref{hatqpw})$ is in turn necessitated by our application of~$(\ref{byourapp})$.)

A slightly unpleasant feature of these polynomial decaying expressions is that, if
left bare, they
would lead to vector fields $Y$, $\overline{Y}$ which fail to be $C^1$ at $\mathcal{H}^+$,
$\overline{\mathcal{H}}^+$, respectively. This would not in fact pose a problem for the analysis
here. We prefer, however, to introduce a cutoff $\eta_{-1}$ in
$(\ref{alphadef})$--$(\ref{betadef2})$ to emphasize the geometric nature of all
objects involved in the proof.

\begin{proposition}
\label{polukrisimo}
For $r_1^*$ sufficiently small and $R_1^*$ sufficiently large, depending only on $M$,
$\Lambda$ we have
\[
C K^Y(\phi) \ge \hat{q}(\phi), \qquad C K^{\overline{Y}}(\phi) \ge \hat{q}(\phi)
\]
in $r\le r_1$
and $r\ge R_1$, respectively.
\end{proposition}

\begin{proof}
We give the proof only for $K^Y$.  The proof for $K^{\overline Y}$ is similar.

Recall the 
functions $\alpha$, $\beta$ from Section~\ref{Ysec}. 
Let us denote by 
\[
\gamma_b\doteq 2r_b ^{-2}(2M/{r_b}^2 -2\Lambda r_b/3)^{-1}.
\]

In the region $r \le r_1$,  we have
$$
\alpha=2-\mu +\eta_{-1}(r^*) |r^*|^{-\delta},\qquad
\beta= \gamma_b (1-\mu + \eta_{-1}(r^*) |r^*|^{-\delta}).
$$
As a consequence, for all $r_1^*<0$,
\[
\alpha'= \left(\frac {2M}{r^2} -\frac 23 \Lambda r\right)
(1-\mu)+\eta_{-1}' |r^*|^{-\delta}
+\delta \eta_{-1}' |r^*|^{-\delta},
\]
\[
\beta' = \gamma_b \left(\frac {2M}{r^2} -\frac 23 \Lambda r\right)
(1-\mu)+\gamma_b \eta_{-1}' |r^*|^{-\delta}
+\gamma_b \delta \eta_{-1}' |r^*|^{-1-\delta}.
\]
Recall that 
$$
\left(\frac {2M}{r^2} -\frac 23 \Lambda r\right)\Big|_{r=r_b}>0
$$
and note in addition that $\eta'_{-1}\ge 0$ in the region $r\le r_1$. 

Recall now the
expression~$(\ref{toexoumeauto})$ for $K^Y$.
In the region $r\le r_1$, we compute
\begin{eqnarray*}
\alpha \left(\frac{2M}{r^2}-\frac{2\Lambda r}3\right) - \alpha'&=& (1+\eta_{-1} |r^*|^{-\delta})
\left(\frac {2M}{r^2} -\frac 23 \Lambda r\right)-\eta_{-1}' |r^*|^{-\delta}-\delta \eta_{-1} |r^*|^{-1-\delta},\\
\frac {\alpha'}{1-\mu} -\frac {\left(\beta(1-\mu)\right)'}{1-\mu}&=&
(\mu +\gamma_b(1-\mu)-\gamma_b\eta_{-1}(r^*)|r^*|^{-\delta})  \left(\frac {2M}{r^2} -\frac 23 \Lambda r\right)\\
&&\hbox{}
+((1-\mu)^{-1}-\gamma_b)(\eta_{-1}'|r^*|^{-\delta}+\delta\eta_{-1}|r^*|^{-1-\delta}),\\
\frac 1r \left(\frac \alpha{1-\mu} -\beta\right)
&=& r^{-1}\Big(1-\mu)^{-1} (1+\eta_{-1}(r^*)|r^*|^{-\delta-1})+ 1\\
&&\hbox{} -\gamma_b(1- \mu)  -
\gamma_b\eta_{-1}(r^*)|r^*|^{-\delta}\Big).
\end{eqnarray*}
Note that 
\begin{equation}
\label{limitingineq}
r^{-1}(1-\mu)^{-1} \le    \frac14 (2M/r^2 -2\Lambda r/3) (1-\mu)^{-2}  + 
r^{-2}(2M/r^2 -2\Lambda r/3)^{-1} .
\end{equation}
It follows that $r_1$ can be chosen sufficiently close to $r_b$, where the choice depends
only on $M$, $\Lambda$, such that, in the region $r\le r_1$, we have
\begin{equation}\label{eq:KY}
CK^Y(\phi)\ge \left (\frac{(\partial_u\phi)^2}{(1-\mu)^2}+\left(1 + \frac {\eta_{-1}}{1-\mu} |r^*|^{-\delta-1}
\right)
\left({(\partial_v\phi)^2}+|\nabb\phi|^2\right)\right).
\end{equation}
In deriving~\eqref{eq:KY},
we have used the Cauchy-Schwarz inequality to bound the $\partial_u\phi\partial_v\phi$ term.
(The fact that the constants work out is ensured by the limiting inequality
$(\ref{limitingineq})$. It is here that the presence of the 
constant factor $\gamma_b$ in the definition of
$\beta$ is paramount.)

It now follows from \eqref{eq:qphi} that, in the region $r\le r_1$,
$$
C K^Y(\phi)\ge \hat q(\phi),
$$
as desired.

Note finally, that, in conformance with our conventions of Section~\ref{depcon},
the constant $C$ above does not depend on $r_0$, $R_0$, despite the appearance
of the cutoff $\eta_{-1}$. 

\end{proof}

\emph{
Henceforth, let $r_1$, $R_1$ be chosen so that the conclusion of the above
proposition holds.}

The special values $r_1$, $R_1$, $r_2$, $R_2$ now being fixed, constants denoted $C$ now
depend only on $M$, $\Lambda$.

\section{The current $J^{\Theta}$}
\label{pragmaux}
Let $\zeta:[0,1]\to[0,1]$ be a cutoff function such that
\begin{align*}
&\zeta(x)=1,\qquad x\ge 3/4,\\
&\zeta(x)=0,\qquad x\le 1/2.
\end{align*}
Define
\[
\zeta_{(t_1,t_2)}(\tau)= \zeta((\tau-t_1)/(t_1-t_2)).
\]
Let $\theta$ be the Heaviside step-function and let $\Theta$ to be the vector field
\begin{eqnarray*}
\Theta&\doteq& 
\left (\theta(r_1^*-r^*)\zeta_{(t_1,t_2)}(2v-r_1^*)+\theta(r^*-r_1^*)\theta(R_1^*-r^*)
\zeta_{(t_1,t_2)}(t) \right.\\
&&\left.\hbox{}+ 
\theta(r^*-R_1^*)\zeta_{(t_1,t_2)}(2u+R_1^*)\right) T,
\end{eqnarray*}
and define
as before the currents $J^{\Theta}$ and
$K^{\Theta}$. Despite the appearance of the Heaviside function, $\Theta$ is a $C^{0,1}$ 
vector field. It is of the form $\Theta=\xi T$, where $\xi$ is a
spacetime cut-off function
adapted to the $C^{0,1}$ foliation $\Sigma_t$, constant on each leaf $\Sigma_t$. We have
\begin{align*}
K^\Theta(\phi)&=T_{\mu\nu}(\phi)\pi_\Theta^{\mu\nu}\\
	&=-\frac{1}{2 (1-\mu)} \left ( (\theta(r_1^*-r^*)\zeta'_{(t_1,t_2)}(2v-r_1^*)
	\left ((\pa_u\phi)^2 +\frac 12 (1-\mu)|\nabb\phi|^2\right) \right.\\
	&\left.+2\theta(r^*-r_1^*)\theta(R_1^*-r^*)\zeta_{(t_1,t_2)}(t) 
	\left ((\pa_t\phi)^2+(\pa_{r^*}\phi)^2+ (1-\mu) 
	|\nabb\phi|^2\right)\right.\\&\left.+ 
\theta(r^*-R_1^*)\zeta'_{(t_1,t_2)}(2u+R_1^*)\left ((\pa_v\phi)^2 +\frac 12 (1-\mu)|\nabb\phi|^2\right) \right).
\end{align*}

\section{Proof of the main estimates}
\label{proofME}
\subsection{Auxiliary integral quantities}
\label{auxq}
Let us define
the auxiliary quantities
\begin{eqnarray*}
{\bf Q}^x_{i,\phi}(t_1,t_2)&\doteq& \int_{\mathcal{R}(t_1,t_2)} q^x_i(\phi)\\
\hat {\bf Q}_{\phi}(t_1,t_2)&\doteq& \int_{\mathcal{R}(t_1,t_2)}{\hat {q}(\phi)}\\
{\bf F}^{T}_{\phi}(t_1,t_2)&\doteq&
\int_{\mathcal{H}^+\cap \{v_1\le v \le v_2\} }J^T(\phi) n^\mu
+ \int_{\overline{\mathcal{H}}^+\cap \{u_1\le u \le u_2\} }J^{T}_\mu(\phi) n^\mu\\
{\bf Z}^{T}_\phi(t_i)&\doteq&
\int_{\Sigma_{t_i}} J_\mu^Tn^\mu
\end{eqnarray*}
where $q_i^x(\phi)$ is as defined in Section~\ref{apdpq},
with $x=\emptyset, a,b,\ldots$, and $\hat{q}(\phi)$ is as defined
in Section~\ref{errcontrol}.

\subsection{Preliminary inequalities}
\subsubsection{Auxiliary quantity inequalities}
For the boundary terms, 
note first that on the support of $\eta_1$,
$$
n=\frac 1{\sqrt{1-\mu}} T.
$$
On the other hand, in the expression \eqref{eq:N},
each term
gives a nonnegative
contribution to $J_\mu^N n^\mu$, and we have on $\Sigma_t$
\begin{equation}\label{eq:JN}
J_\mu^N n^\mu\approx  {\chi_1}\frac {({\pa_u}\phi)^2}{1-\mu}+
{\overline\chi_1}\frac {({\pa_v}\phi)^2}{1-\mu}  + (\pa_t\phi)^2 + (\pa_{r^*}\phi)^2
+(1-\mu) |\nabb\phi|^2.
\end{equation}
As a consequence,
\begin{equation}
\label{note2}
{\bf Z}^{\tilde{N},P}_\phi(t_i)\le C\,{\bf Z}^{N}_\phi(t_i),
\end{equation}
and
\begin{equation}
\label{note3}
\,{\bf Z}^{T}_\phi(t_i)\le C\,{\bf Z}^N_\phi(t_i).
\end{equation}

\subsubsection{Boundary term inequalities for currents $J$}
In this section we address questions of size of the boundary terms generated by the 
currents $J^X, J^{X^a}, J^{X^b}, J^{X^c}, J^{X^d}$ from Section~\ref{X} and $J^Y, J^{\overline Y}$ 
from Section~\ref{Ysec}.
\begin{proposition}
\label{Bxs}
For $J= J^{X,i}(\phi)$, $J^{{X}^a,3}(\phi)$, $J^{{X}^b,2}(\phi-\phi_0)$, $J^{{X}^b,0}(\phi)$,
$J^{X^c,0}(\phi_0)$, $J^{X^d,4}(\phi-\phi_0)$ we have
\[
\left| \int_{\Sigma_{t}}J_\mu n^\mu\right| \le E\,{\bf Z}^{T}_\phi(t),
\]
\[
\int_{\mathcal{H}^+\cap \{v_1\le v \le v_2\} }J_\mu n^\mu
+
\int_{\overline{\mathcal{H}}^+\cap \{u_1\le u \le u_2\} }J_\mu n^\mu
\le
E\, {\bf F}^T_{\phi}(t_1,t_2) .
\]
\end{proposition}
\begin{proof}
We shall here only prove the proposition in a representative case of the current 
$$
J^{X,2}(\phi)=J^{X_0,0}(\phi_0) +\sum_{\ell\ge 1} J^{X_\ell,2}(\phi_\ell),
$$
defined in Section \ref{Aux}.

Recall that 
$$
{\bf Z}^{T}_\phi(t)=\int_{\Sigma_{t}} T_{\mu\nu}(\phi) T^\mu n^\nu. 
$$
Since on the space-like part of $\Sigma_t$, we have $n=(1-\mu)^{-\frac 12}{T}$,
it follows that there 
$$
T_{\mu\nu}(\phi) T^\mu n^\nu=\frac 1{\sqrt{1-\mu}} T_{\mu\nu}(\phi) T^\mu T^\nu =
\frac 1{4\sqrt{1-\mu}}\left (T_{uu}(\phi) + 2 T_{uv}(\phi) + T_{vv}(\phi)\right),
$$
since in addition $T=1/2(\pa/\pa u+\pa/\pa v)$. On the other hand, since
$n=\pa/\pa u$ and $n=\pa/\pa v$
on the null segments $v$=const and $u$=const of $\Sigma_t$, respectively,
\[
T_{\mu\nu}(\phi) T^\mu n^\nu=
T_{\mu\nu}(\phi) T^\mu \left(\frac {\pa}{\pa u}\right)^\nu=\frac 12 \big(T_{uu}(\phi) + T_{uv}(\phi)
\big),
\]
\[
T_{\mu\nu}(\phi)T^\mu n^\nu=
T_{\mu\nu}(\phi) T^\mu \left(\frac {\pa}{\pa v}\right)^\nu=\frac 12 \big(T_{vv}(\phi) + T_{uv}(\phi)
\big).
\]
Since the space-like portion of $\Sigma_t$ corresponds to $r$-values $r_1\le r\le  R_1$ we have that
$$
{\bf Z}^{T}_\phi(t)\approx \int_{\Sigma_{t}} \left (T_{uv}(\phi) + (1-\overline\chi_1) T_{uu}(\phi) + 
(1-\chi_1) T_{vv}(\phi)\right).
$$
Therefore,
\begin{equation}\label{eq:ZT}
{\bf Z}^{T}_\phi(t)\approx \int_{\Sigma_{t}} \left ((1-\overline\chi_1) (\pa_u\phi)^2+ 
(1-\chi_1) (\pa_v\phi)^2 + (1-\mu) |\nabb\phi|^2\right).
\end{equation}
On the other hand,
\begin{align*}
{\bf F}^{T}_{\phi}(t_1,t_2)&=
\int_{\mathcal{H}^+\cap \{v_1\le v \le v_2\} }J^T(\phi) n^\mu
+ \int_{\overline{\mathcal{H}}^+\cap \{u_1\le u \le u_2\} }J^{T}_\mu(\phi) n^\mu\\ &=
\frac 12\int_{\mathcal{H}^+\cap \{v_1\le v \le v_2\} } (\pa_v\phi)^2
+ \frac 12\int_{\overline{\mathcal{H}}^+\cap \{u_1\le u \le u_2\} } (\pa_u\phi)^2.
\end{align*}

We compare now \eqref{eq:ZT} with the expression for 
$\int_{\Sigma_{t}}J^{X,2}_\mu n^\mu$. Start with the current 
$$
J_\nu^{X_0,0}(\phi_0)=-\frac 1 {r^2} T_{\mu\nu}(\phi_0) \left(\frac{\pa}{\pa r^*}\right)^\mu.
$$
Since $\pa/\pa r^*=1/2(\pa/\pa v-\pa/\pa u)$, we infer that on $\Sigma_t$
\begin{eqnarray*}
\int_{\Sigma_t} |J_\nu^{X_0,0}(\phi_0) n^\nu|&\le& C 
\int_{\Sigma_{t}} \left ((1-\overline\chi_1) (\pa_u\phi_0)^2+ 
(1-\chi_1) (\pa_v\phi_0)^2 + (1-\mu) |\nabb\phi_0|^2\right)\\
& \le& C\, {\bf Z}^{T}_{\phi_0}(t).
\end{eqnarray*}
For $\ell\ge1$, consider now
$$
J_\mu^{X_\ell,2}(\phi_\ell)=-f_\ell T_{\mu\nu}(\phi_\ell) \left ({\pa r^*}\right)^\mu - 
\frac 14 \left(f_\ell'+2\frac{1-\mu}r f_\ell\right)\pa_\mu \phi_\ell^2 + 
\frac 14 \pa_\mu \left(f_\ell'+2\frac {1-\mu}r f_\ell\right)
\phi_\ell^2.
$$
Recall that the functions $f'_\ell(r^*)=1$ on the interval $[r_0^*, R_0^*]$ and 
vanish for $r^*\le 2r^*_0$ and $r^*\ge 2R^*_0$. On the space-like portion of $\Sigma_t$, i.e., for 
$r\in [r_1,R_1]$, we have 
$$
|J_\mu^{X_\ell,2}(\phi_\ell) n^\mu|\le C\left (T_{uv}(\phi_\ell) +T_{uu}(\phi_\ell)+T_{vv}(\phi_\ell)
+\phi_\ell^2\right),
$$
while on the null segments of $\Sigma_t$, we have
$$
\left|J_\mu^{X_\ell,2}(\phi_\ell) \left(\frac{\pa}{\pa v}\right)^\mu\right|\le E
 \left (T_{uv}(\phi_\ell) +T_{uu}(\phi_\ell)+ (\pa_u\phi_\ell)^2
+(1-\mu) \phi_\ell^2\right)
$$
for $v={\rm const}$,
and 
$$
\left|J_\mu^{X_\ell,2}(\phi_\ell) \left(\frac{\pa}{\pa u}\right)^\mu\right|\le E
 \left (T_{uv}(\phi_\ell) +T_{vv}(\phi_\ell)+ (\pa_u\phi_\ell)^2
+(1-\mu) \phi_\ell^2\right)
$$
for $u={\rm const}$. Therefore,
\begin{align*}
\int_{\Sigma_t} |J_\mu^{X_\ell,2}(\phi_\ell) n^\mu|&\le E
\int_{\Sigma_t}  \left (T_{uv}(\phi_\ell) + (1-\overline\chi_1) T_{uu}(\phi_\ell) + 
(1-\chi_1) T_{vv}(\phi_\ell)+(1-\mu) \phi_\ell^2\right)\\&\le E
\int_{\Sigma_t}  \left ( (1-\overline\chi_1)(\pa_u\phi_\ell)^2 + 
(1-\chi_1) (\pa_v\phi_\ell)^2+(1-\mu) (|\nabb\phi_\ell|^2+\phi_\ell^2)\right).
\end{align*}
Summing over $\ell$ and using the identity 
$$
\ell(\ell+1)\int_{\mathbb S^2} \frac{\phi_\ell^2}{r^2} \, r^2 dA_{\mathbb S}^2
=  \int_{\mathbb S^2} |\nabb\phi_\ell|^2 \,r^2 dA_{\mathbb S}^2, 
$$
we obtain 
$$
\int_{\Sigma_t} |J_\mu^{X,2}(\phi) n^\mu|\le 
E\int_{\Sigma_t}  \left ( (1-\overline\chi_1)(\pa_u\phi)^2 + 
(1-\chi_1) (\pa_v\phi)^2+(1-\mu) |\nabb\phi|^2\right)\le E\,{\bf Z}^{T}_{\phi}(t),
$$
as desired. On the other hand,
$$
\int_{\mathcal{H}^+\cap \{v_1\le v \le v_2\} }J^{X,2}_\mu(\phi) n^\mu=
\int_{\mathcal{H}^+\cap \{v_1\le v \le v_2\} }\left (J^{X_0,0}_\mu(\phi_0) +\sum_{\ell\ge 1}
J_\mu^{X_\ell,2}(\phi_\ell)\right) \left (\frac{\pa}{\pa v}\right)^\mu.
$$
On the horizon ${\mathcal H}^+$,
 we have $(1-\mu)=0$, and thus $T_{uv}=0$, 
$f_\ell'=f_\ell''=0$ and 
\begin{align*}
J^{X_0,0}_\mu(\phi_0) \left (\frac{\pa}{\pa v}\right)^\mu=&-\frac 1{r^2} T_{\mu\nu}(\phi_0) \left (\frac{\pa}{\pa v}\right)^\mu= -\frac 1{2r^2} (\pa_v\phi_0)^2,\\ 
J^{X_\ell,2}_\mu(\phi_\ell) \left (\frac{\pa}{\pa v}\right)^\mu=&
\left (-f_\ell T_{\mu\nu}(\phi_\ell) \left ({\pa r^*}\right)^\mu - 
\frac 14 (f_\ell'+2\frac{1-\mu}r f_\ell)\pa_\mu \phi_\ell^2 \right.\\&\left.+ \frac 14 \pa_\mu (f_\ell'+2\frac {1-\mu}r f_\ell)
\phi_\ell^2\right) \left (\frac{\pa}{\pa v}\right)^\mu=-f_\ell  (\pa_v\phi_\ell)^2.
\end{align*}
As a consequence, since $|f_\ell(r^*)|\le E$,
\begin{align*}
\left|\int_{\mathcal{H}^+\cap \{v_1\le v \le v_2\} }J^{X,2}_\mu(\phi) n^\mu \right|&\le
\int_{\mathcal{H}^+\cap \{v_1\le v \le v_2\} }\left (\frac 1{2r^2} (\pa_v\phi_0)^2 + 
\sum_{\ell\ge 1} |f_\ell| (\pa_v\phi_\ell)^2\right)\\ &\le 
E \int_{\mathcal{H}^+\cap \{v_1\le v \le v_2\} } (\pa_v\phi)^2\le E \,{\bf F}^{T}_{\phi}(t_1,t_2),
\end{align*}
as desired. Similar arguments give the inequality on the horizon $\overline{\mathcal H}^+$,
as well as the inequalities for the other
currents of the statement of the proposition.
\end{proof}

\begin{proposition}
\label{posvan}
For $J= J^Y(\phi), J^{\overline{Y}}(\phi)$, we have
\[
0\le  \int_{\Sigma_{t}}J_\mu n^\mu \le C\, {\bf Z}^{\tilde{N},P}_\phi(t),
\]
Moreover
\[
\int_{\mathcal{H}^+\cap \{v_1\le v \le v_2\} }J_\mu^Y(\phi) n^\mu\ge 0,
\]
\[
\int_{\mathcal{H}^+\cap \{v_1\le v \le v_2\} }J_\mu^{\overline{Y}}(\phi) n^\mu=0,
\]
\[
 \int_{\overline{\mathcal{H}}^+\cap \{u_1\le u \le u_2\} }J^Y_\mu(\phi) n^\mu=0,
\]
\[
 \int_{\overline{\mathcal{H}}^+\cap \{u_1\le u \le u_2\} }J^{\overline{Y}}_\mu(\phi) n^\mu\ge 0.
\]
\end{proposition}
\begin{proof}
Once again, we shall here consider only the current $J^Y(\phi)$. The considerations
for $J^{\overline{Y}}$ are practically identical.

 By 
the construction in
Section~\ref{Ysec}, the support of $J^Y(\phi)$ is contained in the region $r^*\le \frac 12 r_1^*$. 
This immediately implies that $J^Y(\phi)|_{\overline{\mathcal H}^+}=0$.
Moreover
$Y$ is a bounded future-directed time-like vector field in the region $r_1^*\le r^*< \frac 12 r_1^*$, which implies
that there we have 
\begin{align*}
0\le J_\mu^Y(\phi) n^\mu&= T_{\mu\nu}(\phi) Y^\nu n^\mu\\
&\le C T_{\mu\nu}(\phi) T^\nu n^\mu\\ &=
C J_\mu^T(\phi) n^\mu\le C J_\mu^{\tilde N}(\phi) n^\mu.
\end{align*}
Restricted to  the support of $Y$, the remaining part of $\Sigma_t$ is contained in a 
null segment $v={\rm const}$. Thus, we have 
\begin{align*}
J_\mu^Y(\phi) n^\mu&=T_{\mu\nu}(\phi)Y^\nu \left (\frac{\pa}{\pa v}\right)^\mu=
\frac {\alpha}{1-\mu} T_{uu} + \beta T_{uv}\\ 
&=\frac{\alpha}{1-\mu} (\pa_u\phi)^2 + \beta (1-\mu) |\nabb\phi|^2.
\end{align*}
The functions $\alpha, \beta$ are non-negative and in the region $r\le r_1$ are given by 
$$
\alpha=2-\mu +\eta_{-1}(r^*) |r^*|^{-\delta},\qquad \beta=\gamma_b(1-\mu +\eta_{-1}
(r^*) |r^*|^{-\delta}),
$$
which implies that 
$$
0\le J_\mu^Y(\phi) n^\mu\le C\left (\frac{ (\pa_u\phi)^2}{1-\mu} + (1-\mu) |\nabb\phi|^2\right).
$$
Comparing this to the expression for $J_\mu^N(\phi) n^\mu$, given in \eqref{eq:JN}, we
see that 
$$
0\le J_\mu^Y(\phi) n^\mu\le C J_\mu^N(\phi) n^\mu
$$
on the null portion $v={\rm const}$
 of $\Sigma_t$. Since on this portion we have $N=\tilde N$, we finally obtain the
desired inequality 
$$
0\le \int_{\Sigma_t} J_\mu^Y(\phi) n^\mu \le C\,  {\bf Z}^{\tilde{N},P}_\phi(t).
$$
On the other hand, on ${\mathcal H}^+$ we have $n=\frac{\pa}{\pa v}$ and 
$$
J_\mu^Y(\phi) n^\mu=\frac {\alpha}{1-\mu} T_{uv} + \beta T_{vv}=2|\nabb\phi|^2,
$$
since $\beta=(1-\mu)=0$ and $T_{uv}=(1-\mu)|\nabb\phi|^2$ on ${\mathcal H}^+$. Thus 
$$
\int_{\mathcal{H}^+\cap \{v_1\le v \le v_2\} }J_\mu^Y(\phi) n^\mu=2\int_{\mathcal{H}^+\cap \{v_1\le v \le v_2\} } |\nabb\phi|^2.
$$
\end{proof}

\subsection{Applications of the integral identity for currents}
In this section we exploit the divergence theorem for compatible currents to relate
various integral quantities.
From Proposition~\ref{Bxs}, Proposition~\ref{posvan} and
identity $(\ref{conid})$, the following two propositions follow immediately:
\begin{proposition}
\label{stb}
For $J=  J^{X,i}(\phi)$, $J^{{X}^a,3}(\phi)$, $J^{{X}^b,2}(\phi-\phi_0)$, $J^{{X}^b,0}(\phi)$,
$J^{X^c,0}(\phi_0)$, $J^{X^d,4}(\phi-\phi_0)$ and 
$K=\nabla^\mu J_\mu$,  we have
\[
\left| \int_{\mathcal{R}(t_1,t_2)} K\,\,\right| \le E({\bf Z}^{T}_\phi(t_1)+{\bf Z}^{T}_\phi(t_2)
+{\bf F}^T_\phi(t_1,t_2)).
\]
\end{proposition}
\begin{proposition}
\label{oldIYs}
We have
\begin{equation}
\label{ineq1}
\int_{\mathcal{R}(t_1,t_2)} K^Y(\phi)+ \int_{\Sigma_{t_2}}J^Y_\mu(\phi) n^\mu
+
\int_{\mathcal{H}^+\cap \{v_1\le v \le v_2\} }J_\mu^Y(\phi) n^\mu
\le  C\,{\bf Z}^{\tilde{N},P}_\phi(t_1),
\end{equation}
\begin{equation}
\label{ineq2}
\int_{\mathcal{R}(t_1,t_2)} K^{\overline{Y}}(\phi) + \int_{\Sigma_{t_2}}J_\mu^{\overline{Y}}(\phi) n^\mu+
\int_{\overline{\mathcal{H}}^+\cap \{u_1\le u \le u_2\} }J^{\overline{Y}}_\mu(\phi) n^\mu
\le C\,{\bf Z}^{\tilde{N},P}_\phi(t_1).
\end{equation}
\end{proposition}

Applying  $(\ref{conid})$ with the energy current $J^T_{\mu}$ gives us the following
proposition
\begin{proposition}
\label{evaveo}
\[
{\bf Z}_\phi^T(t_2)+{\bf F}_{\phi}^T(t_1,t_2) =\, {\bf Z}^{T}_\phi(t_1)\le C\, {\bf Z}^N_\phi(t_1),
\]
\end{proposition}
\begin{proof}
The statement follows immediately from $K^T=0$ and $(\ref{note2})$, $(\ref{note3})$.
\end{proof}

Applying $(\ref{conid})$ with the energy current $J^{T}_{\mu} +J^Y_{\mu}+
J^{\overline{Y}}_{\mu}$, we obtain
\begin{proposition}\label{achronal}
Let $\Sigma'\subset \mathcal{R}_{t_1,t_2}$ be achronal. Then
\[
\int_{\Sigma'} J^N_\mu(\phi)n^\mu \le {\bf Z}^N_\phi(t_1)+
C \int_{\mathcal{R}(t_1,t_2)\cap J^-(\Sigma')} -K^{Y}(\phi)  -K^{\overline{Y}}(\phi)
\]
In particular,
\[
{\bf Z}^N_\phi(t_2) \le {\bf Z}^N_\phi(t_1)+ C\int_{\mathcal{R}(t_1,t_2)} -K^{Y}(\phi)-K^{\overline{Y}}(\phi).
\]
\end{proposition}
\begin{proof}
The vector field $T+Y+\overline{Y}$ is timelike and one sees easily
\[
\left( J^{T}_{\mu}+ J^Y_{\mu} +J^{\overline{Y}}_\mu\right)n^\mu\approx J^N_{\mu} n^\mu,
\]
while certainly $K^{T}+K^Y+K^{\overline{Y}}= K^Y+K^{\overline{Y}}$.  \end{proof}

We have an alternative bound on ${\bf Z}^N_\phi(t_2)$ as follows
\begin{proposition}
\label{privavafora}
\[
{\bf Z}^N_\phi(t_2) \le C\, {\bf Z}^{\tilde{N},P}_\phi(t_2)-C\int_{\mathcal{R}(t_1,t_2)} K^{\Theta}(\phi).
\]
\end{proposition}
\begin{proof}
Recall that by the construction given in Section~\ref{pragmaux}, the vector field $\Theta$ is
future timelike, $\Theta|_{\Sigma_{t_2}}=T$ and  $\Theta|_{\Sigma_{t_1}}=0$. 
The result now  follows from the statement
\[
{\bf Z}^N_\phi(t_2) \le C\, {\bf Z}^{\tilde{N},P}_{\phi}(t_2)+ C\int_{\Sigma_{t_2}} J_\mu^\Theta (\phi) n^\mu
\]
and the divergence theorem.
\end{proof}

\subsection{Bounding ${\bf Q}_\phi$ from ${\bf Z}^N_\phi$}
In this section we establish the first key part of 
Theorem~\ref{ME}, that is to say, 
statement $(\ref{MEstatement})$. We begin with the following 
\begin{proposition}
\label{megalnprot}
\begin{eqnarray*}
{\bf Q}_{1,\phi}(t_1,t_2) &\le& E\, {\bf Z}^N_\phi(t_1)
+
\epsilon\, \hat{\bf {Q}}_{\phi}(t_1,t_2)
\end{eqnarray*}
\end{proposition}
\begin{proof}
The proposition follows from the three Lemmas below:
\begin{lemma}
\label{q1lem}
\[
{\bf Q}^a_{1,\phi}+{\bf Q}^{a'}_{1,\phi}\le C \int_{\mathcal{R}(t_1,t_2)}K^{X,3}(\phi)+
 C\int_{\mathcal{R}(t_1,t_2)}K^{X^a,3}(\phi)+
\epsilon\,\hat{\bf Q}_\phi.
\]
\end{lemma}
\begin{lemma}
\label{q3lem}
\[
{\bf Q}^b_{1,\phi} \le C \int_{\mathcal{R}(t_1,t_2)}K^{X,3}(\phi)+
 C\int_{\mathcal{R}(t_1,t_2)}K^{X^a,3}(\phi)+C  \int_{\mathcal{R}(t_1,t_2)}
K^{X^{b},2}(\phi-\phi_0)+
\epsilon\,\hat{\bf Q}_\phi.
\]
\end{lemma}
\begin{lemma}
\label{q4lem}
\begin{eqnarray*}
{\bf Q}^{d}_{1,\phi} &\le& C \int_{\mathcal{R}(t_1,t_2)}K^{X,3}(\phi)+
 C\int_{\mathcal{R}(t_1,t_2)}K^{X^a,3}(\phi)+C \int_{\mathcal{R}(t_1,t_2)}
K^{X^{b},2}(\phi-\phi_0)\\
&&\hbox{}
+C  \int_{\mathcal{R}(t_1,t_2)}
K^{X^c,0}(\phi_0)+C\int_{\mathcal{R}(t_1,t_2)} K^{X^{d},4}(\phi-\phi_0)+
\epsilon\,\hat{\bf Q}_\phi.
\end{eqnarray*}
\end{lemma}
\begin{proof}The statements of the above three Lemmas follow directly from
\eqref{eq:X3}, \eqref{eq:Xa}, \eqref{eq:Xb} and \eqref{eq:Xd}, combined with the 
observation that, since the supports of the currents $J^{X^a,3}$, $J^{X^b,2}$ and
$J^{X^d,4}$, as well as of the quantities $q_{1}^x(\phi$), are contained in the region 
$\{2r_1^*\le r^*\le 2R_1^*\}$, since $K^{X^c,0}(\phi_0)$ is nonnegative, and since
$K^{X,3}$ is nonnegative in $\{r_0^*\le r^*\le R_0^*\}$,
it suffices to apply Lemma~\ref{controllem}.
\end{proof}
Proposition \ref{megalnprot} now follows from Proposition~\ref{stb} and Proposition~\ref{evaveo}.
\end{proof}
\begin{proposition}
\label{KYboundf}
\[
 \int_{\mathcal{R}(t_1,t_2)} K^{Y}(\phi)+
K^{\overline{Y}}(\phi) \le C\, {\bf Z}^{\tilde{N},P}_\phi(t_1).
\]
\end{proposition}
\begin{proof}
This follows immediately by adding $(\ref{ineq1})$, ($\ref{ineq2})$ of Proposition~\ref{oldIYs}, 
in view also of Proposition~\ref{posvan}.

\end{proof}

\begin{proposition}
\label{eukoloprop}
\[
\hat {\bf {Q}}_\phi(t_1,t_2)\le C\, {\bf Q}_{1,\phi} (t_1,t_2)+ C\int_{\mathcal{R}(t_1,t_2)}
K^Y(\phi)+K^{\overline{Y}}(\phi)
\]
\end{proposition}
\begin{proof}
The statement follows from the inequality
\[
\hat{q}(\phi)\le Cq_1(\phi)+ K^Y(\phi)+
K^{\overline{Y}}(\phi),
\]
which is a direct consequence of \eqref{eq:qphi0},
\eqref{eq:qphi} and Proposition \ref{polukrisimo}.
\end{proof}

It now immediately follows from Proposition~\ref{megalnprot},
Proposition~\ref{KYboundf} and Proposition~\ref{eukoloprop} and our conventions
regarding constants $C$, $E$ and $\epsilon$
that the following holds:
\begin{proposition}
\label{teliknepilogn}
If $r_0^*$ is chosen sufficiently small, and $R_0^*$ is chosen
sufficiently large, depending only on $M$, $\Lambda$, then
\[
\hat{\bf Q}_\phi(t_1,t_2)\le E\, {\bf Z}^N_\phi(t_1).
\]
\end{proposition}

\emph{Henceforth, let $r_0$, $R_0$ be so chosen so that the conclusion of the above
proposition holds. In particular, in what follows we shall need only make use of
constants $C$ depending only on $M$, $\Lambda$.}

Finally, since by \eqref{eq:qphi}, \eqref{eq:justqphi}
\[
{\bf Q}_\phi(t_1,t_2)\le \hat{\bf {Q}}_\phi(t_1,t_2),
\]
the statement $(\ref{MEstatement})$ of Theorem~\ref{ME} follows immediately. 

%

\subsection{Bounding ${\bf Z}^N_\phi$ from ${\bf Z}^{\tilde{N},P}_\phi$,  ${\bf Q}_\phi$ and ${\bf Q}_{\Omega_i\phi}$}
Bound $(\ref{MEnew})$ of Theorem~\ref{ME} is contained in the following
\begin{proposition}
\label{metaaux} For all $t_1<t_2$ 
\[
{\bf Z}_\phi^N(t_2)\le C{\bf Z}^{\tilde{N},P}_\phi(t_2) +(t_2-t_1)^{-1} C\Big(
{\bf Z}_\phi^N(t_1)+{\bf Q}_\phi(t_1,t_2)+\sum_{i=1}^3{\bf Q}_{\Omega_i\phi}(t_1,t_2)\Big).
\]
\end{proposition}
\begin{proof}
By Proposition~\ref{privavafora} 
\[
{\bf Z}_\phi^N(t_2)\le C{\bf Z}^{\tilde{N},P}_\phi(t_2) - C\int_{\mathcal{R}(t_1,t_2)} K^\Theta(\phi).
\]
The current $K^\Theta(\phi)$ was defined in Section \ref{pragmaux}. One
easily sees 
$$
 -K^\Theta(\phi)\le \frac C{t_2-t_1}\left (\chi_1 \frac{(\pa_u\phi)^2}{1-\mu} +\overline
 \chi_1 \frac{(\pa_u\phi)^2}{1-\mu} + |\nabb\phi|^2 + \eta_1
 \left((\pa_t\phi)^2+(\pa_{r^*}\phi)^2\right)\right).
$$
Comparing this with \eqref{eq:justqphi}, we infer that the statement of the proposition would
follow from the estimate
\begin{equation}
\label{wouldfollow}
\int_{\mathcal{R}(t_1,t_2)} \eta_2 \left((\pa_t\phi)^2+|\nabb\phi|^2\right)\le C
\Big({\bf Z}_\phi^N(t_1)+
{\bf Q}_\phi(t_1,t_2)+\sum_{i=1}^3{\bf Q}_{\Omega_i\phi}(t_1,t_2)\Big).
\end{equation}

 From \eqref{eq:justqphi} 
we have that
\begin{align*}
&C{\bf Q}_\phi(t_1,t_2)\ge \int_{\mathcal{R}(t_1,t_2)} 
(r-3M)^2|\nabb\phi|^2+(\pa_{r^*}\phi)^2,\\
&C\sum_{i=1}^3{\bf Q}_{\Omega_i\phi}(t_1,t_2)\ge \int_{\mathcal{R}(t_1,t_2)} \eta_1
|\pa_{r^*}\nabb\phi|^2.
\end{align*}
A one-dimensional Poincar\'e inequality immediately implies that 
$$
C\left({\bf Q}_\phi(t_1,t_2)+\sum_{i=1}^3{\bf Q}_{\Omega_i\phi}(t_1,t_2)\right)\ge 
\int_{\mathcal{R}(t_1,t_2)} \eta_2 |\nabb\phi|^2.
$$
It thus follows from \eqref{eq:Xbd02} that
\begin{align*}
\int_{\mathcal{R}(t_1,t_2)} &\left (C K^{X,3}(\phi) +C K^{X^a,3}(\phi) + 
C K^{X^b,2}(\phi-\phi_0) +C K^{X^b,0}(\phi)\right.\\
&\qquad\left.+CK^{X^c,0}(\phi_0)+CK^{X^d,4}(\phi-\phi_0)\right)\\
&\ge \int_{\mathcal{R}(t_1,t_2)}  \eta_2(\partial_t\phi)^2
-C\left({\bf Q}_\phi(t_1,t_2)+\sum_{i=1}^3{\bf Q}_{\Omega_i\phi}(t_1,t_2)\right).
\end{align*}
The desired $(\ref{wouldfollow})$ now follows
from Propositions \ref{stb} and \ref{evaveo}.
\end{proof}

\subsection{Bounding ${\bf Z}^N_\phi(t_2)$ from ${\bf Z}^N_\phi(t_1)$ and ${\bf Q}_\phi$}
The final statements of Theorem~\ref{ME} follows from
\begin{proposition}
\label{xwrisaux}
Let $\Sigma'\subset \mathcal{R}(t_1,t_2)$ be achronal. Then
\begin{equation}
\label{piogevika2}
\int_{\Sigma'} T_{\mu\nu}(\phi)N^\mu n^\nu \le C\Big({\bf Z}^N_\phi(t_1)+
 {\bf {Q}}_\phi(t_1,t_2)\Big).
\end{equation}
In particular, 
\begin{equation}
\label{MEstatement32}
{\bf {Z}}^N_\phi(t_2) \le C\Big({\bf Z}^N_\phi(t_1)+
 {\bf {Q}}_\phi(t_1,t_2)\Big).
 \end{equation}
\end{proposition}
\begin{proof}
By Proposition \ref{achronal}
\[
\int_{\Sigma'} J^N_\mu(\phi)n^\mu \le {\bf Z}^N_\phi(t_1)+
C \int_{\mathcal{R}(t_1,t_2)\cap J^-(\Sigma')} -K^{Y}(\phi)  -K^{\overline{Y}}(\phi)
\]
Recall that the current $K^Y(\phi)$ (respectively, $K^{\overline{Y}}(\phi)$) is positive 
for $r\le r_1$ (respectively, $r\ge R_1$) and  vanishes for $r^*\ge \frac 12 r_1^*$
(respectively, $r^*\le \frac 12 R_1^*$). Moreover, comparing~\eqref{eq:justqphi} 
and~\eqref{toexoumeauto}
easily implies the bound
\[
-K^Y(\phi)-K^{\overline{Y}}(\phi) \le C{q}(\phi)
\]
for $r_1^*\le r^*\le \frac12 r_1^*$, and $\frac12 R_1^*\le r^*\le \frac 12 R_1^*$. The result now follows immediately.
\end{proof}

\section{Proof of Theorem~\ref{gendata}}
\label{proofgendata}
\subsection{Energy decay}   
\label{edsec}
\begin{proposition}
\label{prwtodecay}
There exist constants $C$, $c$, depending only on $M$, $\Lambda$, such that
for all $\phi_\ell$ solutions of $\Box_g\phi_\ell=0$ in $J^+(\Sigma_0)\cap\mathcal{D}$
with spherical harmonic number
$\ell$, then
\[
{\bf Z}_{\phi_\ell}^{{N}} (t)+{\bf Q}_{\phi_\ell} (t,t^*)\le C\,{\bf {Z}}^N_{\phi_\ell} (0)e^{-2ct/\ell^2}
\]
for all $t$ and all $t^*\ge t$.
\end{proposition}
\begin{proof}
This follows immediately from estimates
$(\ref{MEstatement})$ and $(\ref{MEnew})$,
in view of the following lemma, proved in Appendix~\ref{calcap},
applied to the functions $f(t)=\hat{\bf Z}_{\phi_\ell}(t)$ and $h(t)={\bf Z}_{\phi_\ell}^{{N}} (t)$:
\begin{lemma}
\label{calculus}
Let $k$, ${k}_0$ be positive constants and
let $f:[0,\infty)\to\mathbb R$, $g;[0,\infty)\to \mathbb R$ be nonnegative continuous functions
satisfying 
\[
h(t_2)+\int_{t_2}^{t_3}f(\tau)d\tau\le k \left( f(t_2) +
(t_2-t_1)^{-1}\left (h(t_1)+\ell^2 \int_{t_1}^{t_2} f(\tau)  d\tau\right)\right)
\]
for all $t_3>t_2>t_1\ge 0$,
and $\int_0^\infty f \le k_0$. Then there exists a 
constants $c$ depending only on $k$, and a universal constant $C$ such that
\[
h(t)+\int_{t}^{t^*}f(\bar{t})d\bar{t}\le C(\max\{ h(0),k_0\} )e^{-ct/\ell^2}
\]
for all $t$ and for all $t^*\ge t$.
\end{lemma}

\end{proof}

\begin{proposition}
\label{apauto2}
There exist constants $C$, $c$, depending only on $M$, $\Lambda$, such that
for all $\phi_\ell$ solutions of $\Box_g\phi_\ell=0$ in $J^+(\Sigma_0)\cap\mathcal{D}$
with spherical harmonic number
$\ell$ and for all achronal $\Sigma'\subset \mathcal{D}\cap J^+(\Sigma_0)$,
\[
\int_{\Sigma'}T_{\mu\nu}(\phi_\ell)N^\mu n^\nu
\le  C\,{\bf {Z}}^N_{\phi_\ell} (0)\Big(e^{-2cv_+(\Sigma')/\ell^2}+e^{-2cu_+(\Sigma')/\ell^2}\Big),
\]
in particular
\[
{\bf Z}^N_{\phi_\ell}(t) \le C\,{\bf {Z}}^N_{\phi_\ell} (0)e^{-2ct/\ell^2}.
\]
\end{proposition}
\begin{proof}
This follows immediately from Propositions~\ref{prwtodecay} and the inequality
$(\ref{piogevika})$ of Theorem~\ref{ME}.
\end{proof}

The energy decay statements of Theorem~\ref{second} now follow from the following

\begin{proposition}
Let $\Sigma$ be a Cauchy surface for $\mathcal{M}$, and let $\varphi$,
$\dot\varphi$, ${\bf E}_0(\varphi, \dot\varphi)$ be as in the statement of
Theorem~\ref{gendata} or Theorem~\ref{second}.
Then, there exist  constants $C$, $t_0\ge0$ depending only on $M$, $\Lambda$ and the geometry
of $\Sigma\cap J^-(\mathcal{D})$,  such that
\[
\Sigma_{t_0} \subset J^+(\Sigma)\cap \mathcal{D},
\]
and such that
for all solutions $\phi$ to the wave equation $\Box_g\phi=0$ on $J^+(\Sigma)\cap J^-(\mathcal{D})$,
the estimate
\begin{eqnarray*}
{\bf Z}^N_\phi(t_0)
							 &\le& C\, {\bf E}_0(\varphi, \dot\varphi)
\end{eqnarray*}
holds.
\end{proposition}
\begin{proof}
This is completely standard. We give a sketch to emphasize here too the role of compatible
currents based on vector fields! Extend say $N_0$ from $\Sigma_0$ to 
an arbitrary future-timelike vector field $N$ in $J^-(\Sigma_{t_0})\cap J^+(\Sigma)$,
and consider an arbitrary spacelike foliation $\mathcal{S}_{\tau}$  of this region by manifolds
with boundary, such that $\mathcal{S}_{-1}=\Sigma$, $\mathcal{S}_0 = \Sigma_T$, and
$\partial \mathcal{S}_{\tau}\subset J^-(\Sigma_T)\setminus I^-(\Sigma_T)$.
Consider the analogue of $(\ref{conid})$ in $J^-(\mathcal{S}_\tau)\cap
J^+(\mathcal{S}_{-1})$. Set 
\[
f(\tau)=  \int_{\mathcal{S}_\tau} J_\mu^N n^\mu.
\]
In view of the fact that $N$ is future timelike, 
we obtain
\[
f(\tau) \le f(-1) 
+ \int_{J^-(\mathcal{S}_\tau)\cap
J^+(\mathcal{S}_{-1})} K.
\]
On the other hand, one easily sees that there exists a $C$ depending
on the geometry of our chosen foliation of the compact set $J^-(\Sigma_T)\cap J^+(\Sigma)$
such that for all $\tau\in [-1,0]$, 
\[
\int_{J^-(\mathcal{S}_\tau)\cap
J^+(\mathcal{S}_{-1})} K \le C \int_{-1}^\tau f(\bar{\tau})d\bar\tau.
\]
It now follow that $f(0)\le e^Cf(-1)$.

The result now follows by noting that 
\[
f(-1) \le C\, {\bf E}_0(\varphi,\dot\varphi), \qquad  {\bf Z}^N_\phi(t_0)=  f(0).
\]
\end{proof}

\subsection{Pointwise decay}
The pointwise decay statements follow easily.

\appendix

\section{Proof of Lemma~\ref{vanlem}}
\label{vanapp}
Consider the function 
\[
A(r) \doteq r^3h_\ell(r)=
 \ell(\ell+1)\left(1-\frac {3M}r\right) + \frac {3M}r -\frac {8M^2}{r^2} + 
 \frac {M\Lambda r}3-\frac {2\Lambda^2 r^4}9.
\]
Recall that 
\[
\mu(r) = \frac {2M}r + \frac{\Lambda r^2}3,\qquad \mu^2= \frac {4M^2}{r^2} + \frac {4\Lambda Mr}3 + 
\frac {\Lambda^2 r^4}9.
\]
We may thus rewrite the function as
\[
A(r) = \ell(\ell+1)\left(1-\frac {3M}r\right) + \frac {3M}r + {3\Lambda Mr}- 2\mu^2.
\]
For $\ell\ge 1$,  $A$ is concave in $r$:
\[
\frac{d^2A}{dr^2} = - 6M(\ell^2+\ell-1) r^{-3} - 4 \left(\frac{d\mu}{dr}\right)^2- 
4 \mu \frac{d^2\mu}{dr^2},\qquad \frac{d^2\mu}{dr^2}= 4M r^{-3} +\frac 23 \Lambda>0.
\]

The value of $A$ at $r_{c}$ is given by
\begin{eqnarray*}
A(r_c)&=&  (\ell(\ell+1)-2)\left(1-\frac {3M}{r_c}\right) - \frac {3M}{r_c} + {3\Lambda M r_c}+2- 2\mu^2(r_c)\\
&=&(\ell(\ell+1)-2)\left(1-\frac {3M}{r_c}\right) - \frac {3M}{r_c} + {3\Lambda M r_c},
\end{eqnarray*}
where we have used $\mu(r_{c})=1$.
The same formula evidently holds at $r=r_b$ with $r_b$ replacing $r_c$ everywhere.

Since $r_{c}>3M$, $\ell(\ell+1)\ge 2$ 
and $\Lambda r_{c}^2=3-\frac {6M}{r_{c}}$, we have
\[
A(r_{c})\ge \frac {3M}{r_c} (\Lambda r_{c}^2-1)=  \frac {3M}{r_c} \left(2-\frac {6M}{r_{c}}\right)> 0.
\]
From $r_{b}<3$, we obtain similarly
\[
A(r_{b})\le \frac {3M}{r_b} \left(2-\frac {6M}{r_{b}}\right)<0.
\]
By concavity the function $A(r)$ then has exactly one zero on the interval 
$(r_{b},r_{c})$, and thus, so does $h_\ell(r)= r^{-3} A(r)$.

The second statement of the lemma is clear from the final, which in turn follows
immediately from the form of the function $h_\ell$.

\section{Proof of Lemma~\ref{calculus}}
\label{calcap}
By replacing $k$ with $\max\{k,1\}$, we may assume in what follows that $k\ge 1$.
Let $t_1,\ldots, t_i$ be a sequence with $18k(\ell ^2+1) \ge t_i-t_{i-1} \ge 9k( \ell^2+1)$.
We can choose 
\[
t_{i+1} \in [t_i+ 9k (\ell^2+1), t_i+ 18k (\ell^2+1)]
\]
such that by pigeonhole principle
\begin{equation}
\label{st1}
f(t_{i+1}) \le k^{-1}(9\ell^2+9)^{-1} \int_{t_i+(9\ell^2+9)}^{t_i+(18\ell^2+18)} f (\tau) d\tau. 
\end{equation}
Assumptions of the Lemma then imply that
\[
h(t_i)+\int_{t_i}^{t_{i+1}} f(\tau) d\tau
\le k f({t}_{i}) +\frac 19 h(t_{i-1})+ \frac19 \int_{t_{i-1}}^{t_{i}}f(\tau) d(\tau),
\]
\[
\int_{t_i+(9\ell^2+9)}^{t_{i}+(18\ell^2+18)} f(\tau) d\tau
\le k f({t}_{i}) +\frac 19 h(t_{i-1})+ \frac19 \int_{t_{i-1}}^{t_{i}}f(\tau) d(\tau).
\]
Therefore,
\[
f(t_{i+1})\le (9\ell^2+9)^{-1} \left (f(t_i) +  \frac 1{9k} h(t_{i-1})+
\frac 1{9k} \int_{t_{i-1}}^{t_{i}}f(\tau) d(\tau)\right).
\]
Thus if
\begin{equation}
\label{av1}
f(t_i) \le \bar C 2^{-i-1}
\end{equation}
\begin{equation}
\label{av2}
h(t_{i-1})+\int_{t_{i-1}}^{t_{i}} f(\tau)\, d\tau\le \,\bar C\, k\,  2^{-i+1}
\end{equation}
then
\[
f({t_{i+1}}) \le \bar C( 2^{-i}3^{-2} +  3^{-4}2^{-i+2} )\le \bar C2^{-i-1}
\]
\[
h(t_i)+\int_{t_i}^{t_{i+1}}f(\tau)\,d\tau \le \bar C( k2^{-i-1} + 3^{-2} k 2^{-i+1} )\le \bar C\,k\, 2^{-i}.
\]
Now we have that $(\ref{av1})$, $(\ref{av2})$ indeed hold for $i=1$, where
$t_0=0$, and  $t_1$ is defined so as to satisfy $(\ref{st1})$, with $\bar C=\max\{h(0),  k_0\}$. 
This proves that for the sequence $t_i$, defined above,
$$
h(t_i)+\int_{t_i}^{t_{i+1}} f(\tau)\, d\tau\,\le \bar C\, k\,  2^{-i}\le \bar C\, k\,  e^{-c t_i/k\ell^2}.
$$
Using the assumptions of the Lemma again
we immediately obtain the desired statement.

\end{document}